\RequirePackage{tikz}

\documentclass[pdflatex,sn-mathphys]{sn-jnl}% Math and Physical Sciences Reference Style
\usepackage{enumitem}
\usepackage{bbm}
\usepackage[export]{adjustbox}
\usepackage{subfig}

\usetikzlibrary{arrows,automata}

\newcommand{\N}{{\mathcal{N}}}
\newcommand{\Int}{{\mathbb{N}}}
\newcommand{\Comp}{{\mathcal{C}}}
\newcommand{\ksur}{k^\textup{s}}
\newcommand{\tredist}{t^\textup{r}}

\newcommand{\Bids}{{\mathcal{B}}}
\newcommand{\X}{{\mathcal{X}}}
\newcommand{\U}{{\mathcal{U}}}
\newcommand{\Types}{{\mathcal{T}}}
\newcommand{\D}{{\mathcal{D}}}
\newcommand{\Real}{{\mathbb{R}}}
\newcommand{\reward}{{\zeta}}
\newcommand{\transition}{{\rho}}
\newcommand{\bdist}{{\nu}}
\newcommand{\oprob}{{\gamma}}
\newcommand{\kbar}{{\bar{k}}}
\newcommand{\Oset}{{\mathcal{O}}}
\newcommand{\Prob}{{\mathbb{P}}}
\newcommand{\sd}{{\boldsymbol{d}}}
\newcommand{\epi}{{\boldsymbol{\pi}}}
\newcommand{\Dkbar}{{\D^\kbar}}
\newcommand{\accessfair}{{\text{af}}}
\newcommand{\rewardfair}{{\text{rf}}}
\newcommand{\pbr}{{\tilde{\pi}}}
\newcommand{\floor}[1]{\left\lfloor #1 \right\rfloor}
\newcommand{\ceil}[1]{\left\lceil #1 \right\rceil}

\DeclareMathOperator*{\argmax}{\text{argmax}}
\DeclareMathOperator*{\E}{\mathbb{E}}
\DeclareMathOperator*{\std}{\text{std}}

\def\CENTE{\texttt{DICT}}

\newtheorem{assumption}{Assumption}
\newtheorem{lemma}{Lemma}
\newtheorem{corollary}{Corollary}
\jyear{2022}%

\theoremstyle{thmstyleone}%
\newtheorem{theorem}{Theorem}%  meant for continuous numbers
\newtheorem{proposition}[theorem]{Proposition}% 

\theoremstyle{thmstyletwo}%

\theoremstyle{thmstylethree}%
\newtheorem{definition}{Definition}%

\begin{document}

\title[ ]{A self-contained karma economy for the dynamic allocation of common resources$^\star$}

\author*[1,2]{\fnm{Ezzat} \sur{Elokda}}\email{elokdae@ethz.ch}

\author[1]{\fnm{Saverio} \sur{Bolognani}}\email{bsaverio@ethz.ch}

\author[2]{\fnm{Andrea} \sur{Censi}}\email{acensi@ethz.ch}

\author[1]{\fnm{Florian} \sur{D\"{o}rfler}}\email{floriand@ethz.ch}

\author[2]{\fnm{Emilio} \sur{Frazzoli}}\email{efrazzoli@ethz.ch}

\affil*[1]{\orgdiv{Automatic Control Laboratory}, \orgname{ETH Z\"{u}rich}, \orgaddress{\country{Switzerland}}}

\affil[2]{\orgdiv{Institute for Dynamic Systems \& Control}, \orgname{ETH Z\"{u}rich}, \orgaddress{\country{Switzerland}}}

% \affil[3]{\orgdiv{Department}, \orgname{Organization}, \orgaddress{\street{Street}, \city{City}, \postcode{610101}, \state{State}, \country{Country}}}


\abstract{

This paper presents karma mechanisms, a novel approach to the repeated allocation of a scarce resource among competing agents over an infinite time. Examples include deciding which ride hailing trip requests to serve during peak demand, granting the right of way in intersections or lane mergers, or admitting internet content to a regulated fast channel. We study a simplified yet insightful formulation of these problems where at every instant two agents from a large population get randomly matched to compete over the resource. The intuitive interpretation of a karma mechanism is ``If I give in now, I will be rewarded in the future.'' Agents compete in an auction-like setting where they bid units of karma, which circulates directly among them and is self-contained in the system. We demonstrate that this allows a society of self-interested agents to achieve high levels of efficiency without resorting to a (possibly problematic) monetary pricing of the resource. We model karma mechanisms as dynamic population games and guarantee the existence of a stationary Nash equilibrium. We then analyze the performance at the stationary Nash equilibrium numerically. For the case of homogeneous agents, we compare different mechanism design choices, showing that it is possible to achieve an efficient and ex-post fair allocation when the agents are future aware. Finally, we test the robustness against agent heterogeneity and propose remedies to some of the observed phenomena via karma redistribution.

}




\maketitle

\section{Introduction}
\label{sec:Intro}

The scarcity of resources is one of modern day society's most prominent challenges.
With a strong population growth and shift to urbanization, our finite natural and infrastructure resources are seeing unprecedented levels of stress.
The need to devise \emph{fair} and \emph{efficient} means of access to these resources is now more eminent than ever.

In this paper, we study a class of dynamic resource allocation problems in which an indivisible resource is repeatedly contested between two anonymous users who are randomly drawn from a large population.
Figure~\ref{fig:examples} demonstrates three motivating examples for this class of resource competitions.
\begin{enumerate}[label=(\alph*)]
    \item Due to excessive demand, only one of two trip requests from ride-hailing riders can be served by the closest ride-hailing driver. Which rider should be served? \label{ex:ridehailing}
    
    \item Two autonomous vehicles (AVs) meet at an unsignalled intersection. Which AV should go first? \label{ex:intersections}
    
    \item To improve quality of service for critical content, an internet service provider (ISP) splits its bandwidth into a high capacity fast channel and a low capacity slow channel, and dedicates the fast channel to half of the total traffic. Two internet content providers (ICPs) simultaneously request service. Which ICP should be granted to the fast channel? \label{ex:internet}
    
\end{enumerate}
\begin{figure}[bt!]
    \centering
    \subfloat[]{
        \includegraphics[height=0.3\textwidth, valign=m]{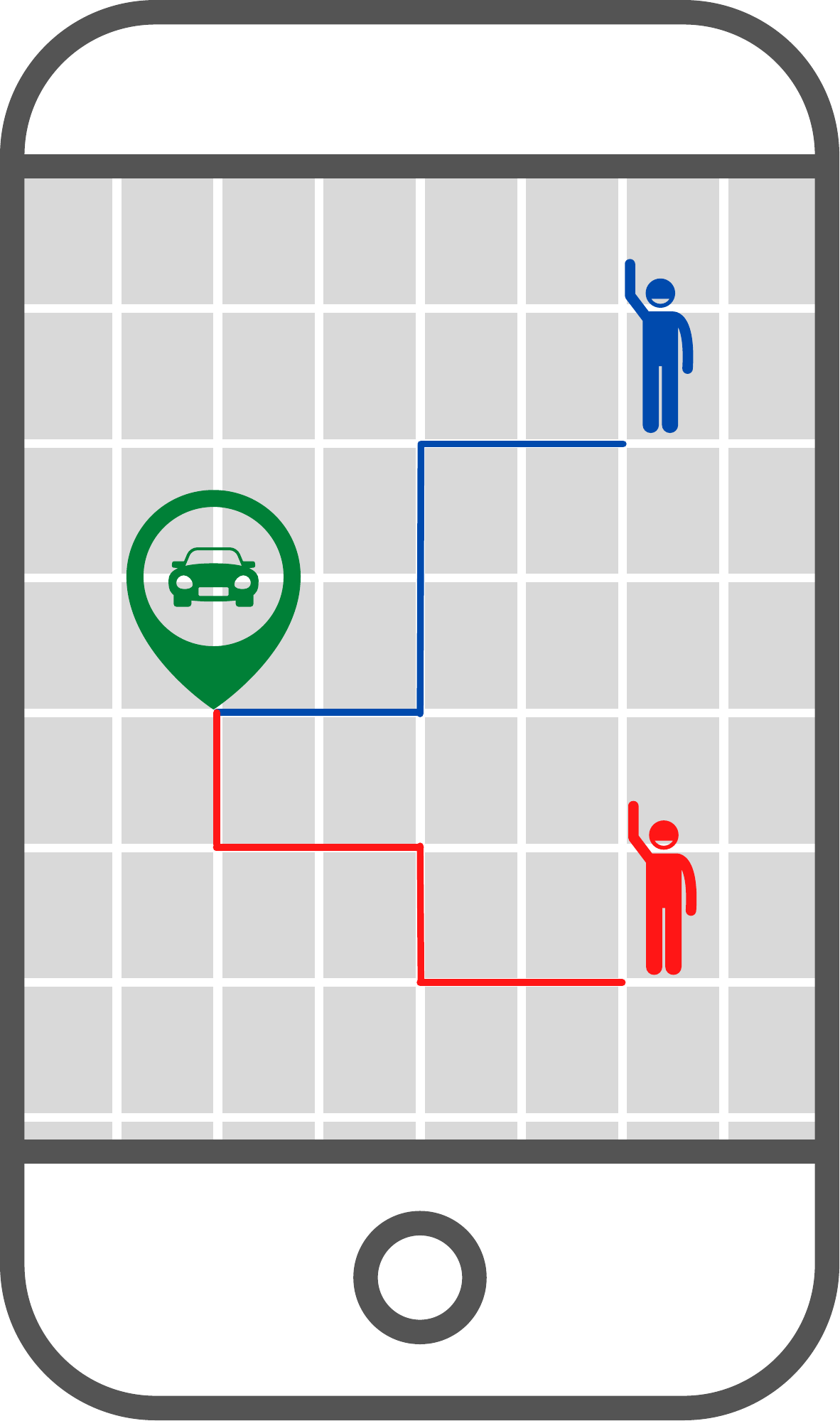}
    }
    \hfill
    \subfloat[]{
        \includegraphics[height=0.3\textwidth, valign=m]{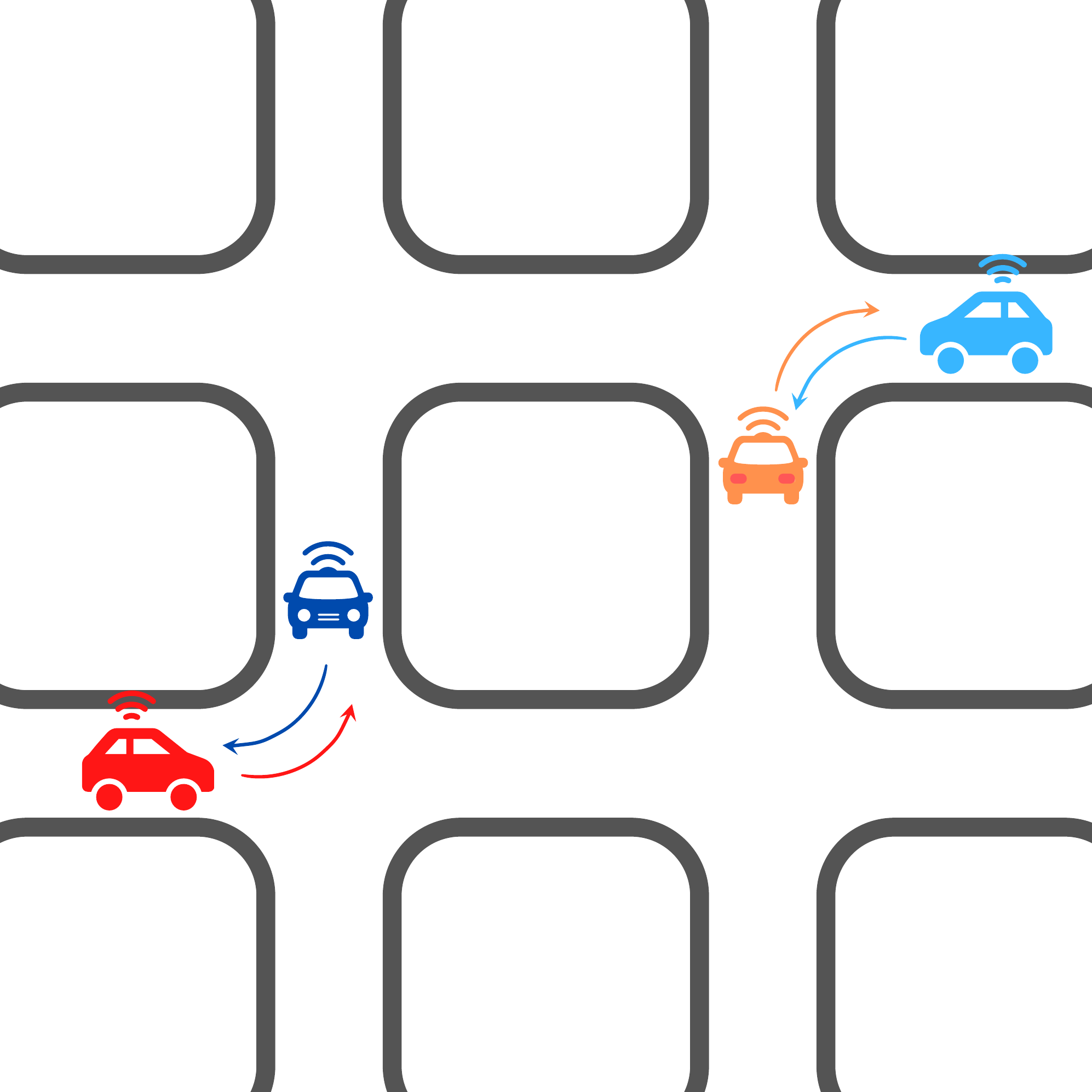}
    }
    \hfill
    \subfloat[]{
        \includegraphics[width=0.38\textwidth, valign=m]{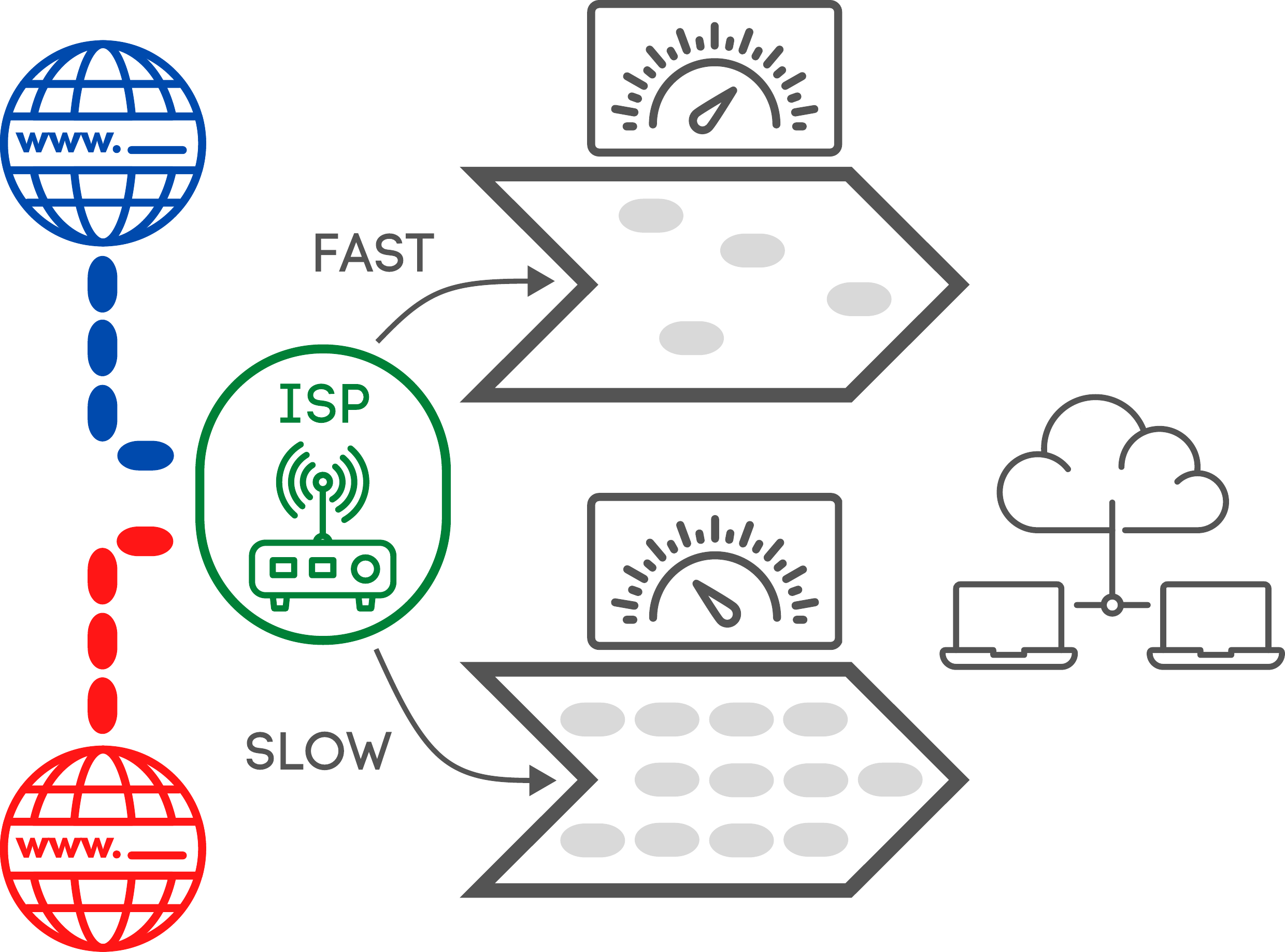}
        \vphantom{\includegraphics[height=0.3\textwidth, valign=m]{example-image-a}}
    }
    
    \caption{Three motivating examples for karma mechanisms.
    (a) When demand exceeds supply, which trip request should a ride-hailing platform assign to the next available driver?
    (b) In a future of autonomous vehicles (AVs), can the AVs self-coordinate who passes first in intersections in a fair and efficient manner?
    (c) Can an internet service provider (ISP) dedicate a fast channel for critical content in a net-neutral manner, i.e., without charging internet content providers (ICPs) differently?}
    \label{fig:examples}
\end{figure}

These examples share in common that the resource (the ride-hailing driver, the intersection, or the fast channel) is repeatedly contested over a long time horizon by anonymous users in a large population (the ride-hailing riders, the AVs, or the ICPs), who could have time-varying private needs for accessing the resource.
It is important to note that the reality of these examples is complex, and they each deserve a separate treatment.
% For example, it is possible that more than two users simultaneously request access to the resource; the ride-hailing driver might prefer a rider whose destination aligns better with an overall travel plan;
% which AV goes first in the intersection might have complex consequences for the overall traffic flow; and a 50-50 split of the fast and slow channels might be sub-optimal for the ISP.
Nonetheless, the level of abstraction chosen in this paper serves to highlight the fundamental trade-offs that arise in this class of problems, as well as to ease the presentation of the novel analysis tools developed to study them.
We adopt a level of generality in these tools that allows them to be readily specified to more complex settings.

% Considering the AVs example, we can take inspiration from how
% }humans often incorporate simple notions of fairness in resolving who goes first in an unsignalled intersection. 
% Using subtle driving maneuvers and possibly hand gestures, human drivers are able to signal their intent and are willing to yield occasionally when we feel that we have been treated generously a few times, even if we are likely intersecting with a different person now than before~\cite{nandavar2019understanding}.
% In some cases this `give and take' behavior is partly codified, e.g., in the all-way stop intersections and in the zipper lane merge.

One can draw inspiration from how small communities sometimes manage to self-organize the access to common resources~\cite{ostrom1990governing}.
Typically these communities devise systems that facilitate a form of fair `giving and taking', i.e., systems where the users take turns in accessing the resource (see, e.g., the system devised by local fishermen to manage the inshore fishery in Alanya, Turkey~\cite{berkes1986local,ostrom1990governing}).
This work is an attempt to systematize the `give and take' such that it can be applied unambiguously in a large-scale system, and builds upon the exploratory concept proposed in~\cite{censi2019today}. The main enabler is \emph{karma}, a counter that encodes the history of `giving and taking' of the users.
Loosely inspired by the notion of karma in Indian tradition~\cite{doniger1980karma}, each user is endowed with karma points which increase when the user yields the resource, and decrease when the user accesses the resource.
A user with a high level of karma was likely yielding in the past, and gets an advantage in receiving the resource now. In turn, the disfavored user that yields now gets compensated in karma which will give them an advantage in the future.

In its simplest form, such a \emph{karma mechanism} is an effective mean to facilitate \emph{turn-taking} on the large scale.
But a karma mechanism can do more than simple turn-taking.
If users are given the option to choose how much karma to use now, e.g., through an auction-like \emph{karma bidding} scheme, the karma becomes also a means to express \emph{private temporal preferences}.
For example, a user may decide to yield (and gain karma) when its urgency is low, in anticipation to the situation where accessing the resource is time-critical.
A karma mechanism can hence facilitate the allocation of the resource to whoever needs it most, i.e., the maximization of resource allocation \emph{efficiency}.

A classical device that is used to express preferences and facilitate access to resources is money.
The ride-hailing platform can raise the trip price until only one of the contesting riders is willing to pay for it -- a common practice referred to as \emph{surge pricing}~\cite{castillo2017surge,cachon2017role}.
This practice has seen some public criticism due to a lack of transparency and a tendency to raise prices in a manner that is deemed unfair~\cite{dholakia2015everyone,shontell2014uber}.
While in principle it is meant to allocate trips to the most needy riders, in practice the trips are allocated simply to those who can afford them.
In the transportation domain, decades of research on the use of monetary road pricing policies has been faced with little public enthusiasm due to concerns for equitable access to the roads~\cite{yang2011managing,xiao2013managing,brands2020tradable}.
A popular remedy is the use of \emph{tradable credits}, which are periodically issued road access tokens that are allowed to be traded in a monetary market.
While these schemes ensure that the yielding users can at least sell their credits and receive (monetary) compensation, they neglect the fact that wealthy users (i.e., those who have high `value of time'~\cite{borjesson2012income}) persist to have a systematic advantage in accessing the resource~\cite{xiao2013managing}.
Finally, the topic of \emph{net neutrality} has seen wide-spread public debate in recent years, with strong concerns that the internet will lose its integrity as an open and free resource if ISPs charge ICPs differently~\cite{obama2016net,bourreau2015net,pil2010net,hahn2006economics}.
In all of the above debates, the potential existence of a simple and efficient non-monetary solution seems to be overlooked.

Karma shares similarities with money in how it acts as a token of exchange, but has the distinguishing feature of being acquired from \emph{fair exchanges} that are relevant to the resource allocation problem at hand.
One need not worry about matters of wealth inequality, or rely on the assumption that money is a universal, objective measure of value, since karma is only acquired from the process of yielding the resource to another user.
Karma hence facilitates the design of a purpose-built, self-contained economy for the resource allocation task.
Like in monetary economies, the karma economy can be \emph{tuned} to achieve different fairness and efficiency objectives in a manner that is targeted to the specific resource allocation task, through the design of karma payment rules, the redistribution of karma, and other techniques that we explore here.

Karma mechanisms promise to be efficient and fair, but these unconventional mechanisms require a novel analysis.
A difficulty in the analysis arises due to the lack of reliance on an extrinsic measure of value. Karma does not have value a-priori, and is never used directly in the cost functions of the users.
The value of karma instead arises from how it facilitates access to the resource, and how users need to ration its use in order to cover their future resource access demands.
This makes the behavior of rational users under the karma mechanism, and the resulting social welfare, difficult to predict, and requires the formulation of non-trivial dynamic games played in large populations.
In this work, and in comparison to~\cite{censi2019today}, we develop a tractable and rigorous game-theoretic model to study karma mechanisms that is built on top of the class of \emph{dynamic population games}~\cite{elokda2021dynamic}, and prove the existence of a suitable notion of equilibrium, the \emph{stationary Nash equilibrium}.
% , and provide analytical tools to study social welfare measures of interest at the equilibrium.
We then utilize these technical tools to numerically investigate the strategic behaviors that emerge under the karma mechanisms, their consequences for the social welfare, as well as how the mechanisms can be tuned to achieve different resource allocation objectives.
% different karma mechanism designs which are tailored to case studies involving different assumptions on the users' private utilities and degree of strategic competence.

\subsection{Related works}

\subsubsection{Repeated games}
\label{subsec:repeatedgames}
The celebrated folk theorem~\cite{friedman1971non,fudenberg1986folk} asserts that any individually rational outcome\footnote{In an individually rational outcome, the payoff of each player weakly dominates the player's security or minimax payoff.
The set of individually rational outcomes include socially efficient outcomes in many games, such as in the prisoner's dilemma.} of a finite-player single-stage game can be sustained in a Nash equilibrium of the infinitely repeated game, provided that the players are sufficiently future aware.
The constructive proofs of this classical result and many of its extensions rely on the notion of \emph{switch strategies}, in which the players initially agree on the socially desirable set of actions to play.
In case a player deviates from the agreed action, all other players effectively punish the deviator by switching to a set of actions that make the deviator worst off.
This requires the ability to both detect a deviation as well as identify the deviator.

Several extensions of the folk theorem consider when the actions of others are not perfectly observable, thereby posing a difficulty in identifying deviators.
These include when the players observe a common public outcome~\cite{fudenberg1994folk} or when they only observe private outcomes~\cite{fudenberg1991approximate}.
These works impose identifiability conditions
% \footnote{In~\cite{fudenberg1994folk}, the condition is \emph{pairwise identifiability} and in~\cite{fudenberg1991approximate} it is \emph{informational connectedness}.}}
on the stage game which essentially guarantee that each player can identify the history of others' actions from the observable outcomes.
Another extension of the folk theorem is to \emph{stochastic games}~\cite{dutta1995folk} where the stage games are time-varying and depend on the previous actions of the players.
The difficulty in this setting is that deviations are not only immediately beneficial, but could also take subsequent games into a regime that is profitable to the deviator on the long run.
The authors impose conditions on the game which essentially guarantee that the long run cost of punishment outweighs the long run benefit of deviation.

All of these works consider that every stage game is played by the same finite set of players.
A more related setting is when two players in a large population are randomly matched in each stage, for which a folk theorem is shown in~\cite{okuno1995social}.
Each player is associated with as social status state that is observable by others.
Deviators are punished by changing their status from good to bad, which all future matched players observe and punish for.
%Therefore, any individually rational outcome can be sustained in a so-called \emph{norm equilibrium}, where in addition to all players playing the desired action the distribution of status states in the population is stationary.
%We remark that the status state is similar to the karma state we consider, and the norm equilibrium is similar to our notion of the stationary Nash equilibrium.\textcolor{blue}{At this point of the paper we have not defined these concepts. I think this sentence can be removed or phrased more generally, saying that these concepts have similarities with the concepts that we will introduce later (but then it's so general that it doesn't help much).}

%Despite these similarities, o
Our setting differs fundamentally from the above cited works.
We consider that the players have time-varying private \emph{preferences}, namely their urgency to acquire a contested resource.
In contrast to folk theorem results with private information~\cite{fudenberg1991approximate}, the privacy is with respect to the payoffs of the opponents, rather than the actions they play.
In contrast to folk theorem results for stochastic games~\cite{dutta1995folk}, the time-varying nature is with respect to the private player preferences rather than a fully observable game state.
In the context of a folk theorem, the socially desirable set of actions in our setting is for the players to report their urgency truthfully such that the resource is allocated to the highest urgency player.
It is not obvious how deviation from truthfulness can be detected in the first place.
In principle, if the time-varying urgency \emph{process} is public and the game is played with the same finite set of players, non-truthfulness can be detected on the long run by correlating the history of each player's reports to the expected history~\cite{jackson2007overcoming}.
% \footnote{To the extent of our knowledge, this setting has not been considered thus far.}
But in a large population setting where maintaining explicit histories is infeasible, we interpret karma as an extension of the social status state in~\cite{okuno1995social} to handle private preferences, essentially placing a budget on how often players can declare high urgency.

\subsubsection{Mechanism design without money}

The famous Gibbard-Satterthwaite impossibility theorem~\cite{gibbard1973manipulation,satterthwaite1975strategy} poses a fundamental challenge in the design of resource allocation mechanisms.
It asserts that when there are three or more alternative allocations, it is impossible to design a strategy-proof mechanism that is non-dictatorial when the domain of preferences is unrestricted.
One avenue to escape this impossibility is in the use of money, which imposes structure on the preferences of the users by measuring them against the objective monetary yardstick.
The problem of designing monetary mechanisms is well studied, with positive results including the Vickrey-Clarke-Groves (VCG) mechanism~\cite{vickrey1961counterspeculation,clarke1971multipart,groves1973incentives}, a general mechanism that is well known to be strategy-proof and lead to efficient allocations.
% However, the use of money is deemed as unfair or unethical in many application domains.
% Examples include the allocation of donated kidneys to transplant patients~\cite{roth2004kidney} and the allocation of college applicants to colleges~\cite{gale1962college}.
On the other hand, the design of \emph{mechanisms without money}~\cite{schummer2007mechanism} is in general more difficult due to the lack of a general instrument that can be used to align incentives.
Some successes include the cases when preferences are single-peaked~\cite{moulin1980strategy}, when each user has one item to trade~\cite{shapley1974cores}, and when matching users pairwise in a bipartite graph~\cite{gale1962college}, which all leverage specific structures in the preferences of the users that are difficult to generalize.
When users must express preferences over many alternatives, a general approach is the \emph{pseudo-market} pioneered by~\cite{hylland1979efficient} and famously adopted in the context of allocating course seats in business schools~\cite{sonmez2010course,budish2011combinatorial,budish2012multi}. In a pseudo-market, users are given a finite budget of tokens to distribute over the alternatives, whose prices (in tokens) are set/discovered to clear the market (i.e., allocate the correct amount of resources to the correct amount of users).
However, pseudo-markets only promise to be Pareto-efficient and are also not strategy-proof (although strategizing becomes difficult when there are many users and alternatives)~\cite{hylland1979efficient}.
It is noteworthy to mention that in our motivating examples, any allocation of the contested resource is Pareto-efficient.

The aforementioned difficulty stems from the fact that the classical mechanism design problem is concerned with a \emph{static} or \emph{one-shot} allocation of goods.
On the other hand, when the allocation is \emph{dynamic} or \emph{repeating over time}, new opportunities for the design of strategy-proof and efficient mechanisms present themselves.
On a conceptual level, just as how money can be used to incentivize truthful behavior (or punish non-truthfulness), a similar incentive can be achieved through a promise of future service (or denial thereof).
Despite of this intuitive notion, the role that repetition could play in mechanism design has only been recently noticed, and the literature on mechanism design for dynamic resource allocation is sparse.
A few recent works build upon the notion of ``promised utilities'', pioneered in the context of contract design in repeated relationships~\cite{spear1987repeated}.
These include~\cite{guo2020dynamic}, which develops an incentive compatible mechanism for the case when a single principal repeatedly allocates a single good to a single user, and~\cite{balseiro2019multiagent}, which extends this approach to the case when the single good is repeatedly allocated to one of the same contesting users.
The working principle of these works is to find a set of future utilities that the principal promises to the user(s) as a function of their reported immediate utility, in a manner that incentivizes truthful reporting, and while ensuring that the principal can recursively keep these future promises.
% In~\cite{guo2020dynamic}, the promised utilities are given the intuition of corresponding to the number of consecutive times the user can claim the good ``with no questions asked''. In~\cite{balseiro2019multiagent}, no intuition is given to how the principal can keep its promises.
This is based on the assumption that the principal has the power to commit to the promised utilities, without specifying the exact mechanism to do so.
Similarly, karma is a device that encodes future promises; the higher a user's karma the more favorable its future position will be.
But with karma, these promises need not be made explicit, and a single principal need not be held accountable for them.
Instead, the future value of karma arises in a decentralized and natural manner as the users strategically ration its use.
The promise of future utility is made by the population as a whole by attributing the right of access to future resources to karma.
% include~\cite{balseiro2019multiagent}, which proposes a mechanism based on ``promised utilities'' that are not intuitive to implement in practice.

In other related works, \cite{sonmez2020incentivized} leverage the high likelihood of kidney transplant failures to incentivize participation in the kidney exchange by providing a priority for re-transplant to participants.
\cite{kim2021organ} similarly incentivize participation in the kidney exchange by issuing participants a voucher for re-transplant that is also redeemable by their offspring.
We consider our karma mechanisms to be complimentary to these works, offering a simple and intuitive alternative that has the potential to scale to large systems and across multiple applications.

\subsubsection{Artificial currency mechanisms}

A special class of mechanisms without money, which are perhaps the most related to our karma mechanisms, are the so-called \emph{artificial currency} or \emph{scrip} mechanisms.
These mechanisms have been proposed in multiple isolated application instances since the early 2000's.
\cite{golle2001incentives} propose a ``point system'' to address the problem of \emph{free-riding} in peer-to-peer networks, where agents tend to download many more files than they upload.
%Uploads are incentivized by rewarding points which are needed to download files.
Similar works in the domain of peer-to-peer networks include~\cite{vishnumurthy2003karma}, who specifically call their point system ``karma'' and focus on its cryptographic implementation rather than the design of the mechanism itself; and~\cite{friedman2006efficiency}, who do incorporate elements of mechanism design but focus solely on the choice of a single parameter, which is the total amount of karma in their specific model.
In the domain of transportation, \cite{salazar2021urgency} recently demonstrated how an artificial currency (also called ``karma'') can be utilized instead of monetary tolls to achieve optimal routing in a two arc road network.
% In the domain of course allocation, the use of artificial currency has been employed in business schools as a means for students to bid over courses with limited seats and to allocate those seats~\cite{sonmez2010course,budish2011combinatorial,budish2012multi}. The course allocation problem is however fundamentally different from the considered class of problems we since it is essentially a one-shot allocation.
To the extent of our knowledge, the only concrete example of a real-life implementation of a karma-like concept is the ``choice system'' for the allocation of food donations to food banks in the United States~\cite{prendergast2022allocation}.
There, food banks are allocated ``shares'' which they use to bid on the food donations they need.
It is considered to be a major success as evidenced by the active participation of food banks in the system as well as the unprecedented fluidity of food donations it resulted in.

All of the above works share in common the need for a non-monetary medium of exchange to coordinate the use of shared resources.
However, there is an apparent lack of unity in the approaches taken, with most works proposing a problem-tailored, heuristic mechanism with little rigorous justification and scope for generalization.
In~\cite{prendergast2022allocation}, a model is presented that makes many simplifying assumptions on the strategic behavior of the users.
This model does not truly capture the dynamic nature of the optimization problem of the users, who must ration their use of shares now to secure their future needs.
In~\cite{salazar2021urgency}, a game-theoretic equilibrium is considered in which individual users solve a finite horizon dynamic optimization, but importantly, the amount of karma saved at the end of the horizon is treated as an exogenous parameter.
A few other works that attempt to systematically study artificial currency mechanisms include~\cite{johnson2014analyzing,gorokh2021monetary}.
\cite{johnson2014analyzing} studies a setting where a pool of users alternate between requesting and providing services to each other (e.g., a pool of parents exchanging baby-sitting services).
This differs from our work in the following fundamental aspects.
First, it does not give the users the flexibility to express intensity of preferences through a
bidding procedure.
Second, the equilibrium notion considered relies on remembering if certain users denied providing service before and punishing those users by never granting them service again.
As discussed in Section~\ref{subsec:repeatedgames}, 
retaliation schemes are crucially based on the capability of detecting defection, which in our setting cannot be done without knowing the private preference of the agents.
%This resembles the classical Folk theorem strategies for punishing deception~\cite{friedman1971non,fudenberg1986folk}, which defeats the purpose of introducing an artificial currency in the first place since it requires public knowledge of the histories and identities of the users and will be impractical to scale to large-scale systems.
\cite{gorokh2021monetary} provides a general method to convert truthful monetary to non-monetary mechanisms, which relies on a central planner estimating how much money each user would spend in a finitely repeated monetary auction and giving the user a similar amount of artificial currency at the beginning of the horizon.
This requires central knowledge of the players’ private
preferences, and does not capture the important dynamic feedback process of gaining currency through yielding.
In contrast to these works, our approach shows that a robust and ultimately efficient behavior emerges only from the dynamic strategic problem faced by the users of the karma mechanism, without additional rules or coordination mechanisms.
To the extent of our knowledge, there are no other works that study the strategic behavior in artificial currency mechanisms at this level of generality.
We believe that this is fundamental for the understanding of these mechanisms and serves as an important tool for the mechanism design.

\subsection{Organization of the paper}

Section~\ref{sec:karma} introduces the setting of dynamic resource allocation and the concept of karma mechanisms.
In Section~\ref{sec:Model} we model karma mechanisms as dynamic population games, and show that a stationary Nash equilibrium is guaranteed to exist.
Section~\ref{sec:KarmaTransferModel} focuses specifically on how different karma payment and redistribution rules can be incorporated in the model.
The model is utilized in a numerical investigation of karma mechanisms in Section~\ref{sec:NumericalAnalysis}, where we provide insights on the emerging strategic behavior as well as the consequences of the karma mechanism design on the achieved efficiency and fairness of the resource allocation.

\subsection{Notation}
Let $D$ be a discrete set and $C$ be a continuous set.
Let $a,d \in D$ and $c \in C$.
For a function $f : D \times C \rightarrow \Real$, we distinguish the discrete and continuous arguments through the notation $f[d](c)$.
Alternatively, we write $f : C \rightarrow \Real^{\lvert D \rvert}$ as the vector-valued function $f(c)$, with $f[d](c)$ denoting its $d^\textup{th}$ element.
Similarly, $g[a \mid d](c)$ denotes the conditional probability of $a$ given $d$ and $c$. 
Specifically, $g[d^+ \mid d](c)$ denotes one-step transition probabilities for $d$.
We denote by $\Delta(D):=\left\{\left. p \in \Real_+^{\lvert D \rvert} \right\rvert \sum_{d \in D} p[d] = 1 \right\}$ the set of probability distributions over the elements of $D$.
%, which has dimension equal to the cardinality of $D$.
For a probability distribution $p \in \Delta(D)$, $p[d]$ denotes the probability of element $d$.
When considering heterogeneous agent types, we denote by $x_\tau$ a quantity associated to type $\tau$.
\section{Karma mechanisms for dynamic resource allocation}
\label{sec:karma}

% \begin{figure}[bt]
% 	\centering
% 	\includegraphics[width=0.65\textwidth]{figures/interactions-Ct.pdf}
% 	\caption{Example dynamic resource allocation problem, where the resource is the right of way in intersections. Anonymous agents in a large population get repeatedly and randomly matched in resource competition instances $\{\Comp_{t_n}\}$. A karma mechanism can be used to select who to allocate the resource to.}
% 	\label{fig:IntersectionExample}
% \end{figure}

% In order to provide an interpretation for these definitions, we refer to the motivating example of \emph{intersection management in a road network} throughout the presentation.
We consider a population of agents $\N = \{1,\dots,N\}$, where the number of agents $N$ is typically large. For example, $\N$ is the set of ride-hailing platform riders in the metropolitan area of interest.

At discrete global time instants $t \in \Int$, two random agents from the population (denoted by $\Comp[t] \subset \N$) compete for a scarce, indivisible resource, such as the closest ride-hailing driver
%right of way in an intersection
.
%(see Figure~\ref{fig:IntersectionExample}).
We are concerned with designing a mechanism that, at each interaction time $t$, selects one of the two agents in $\Comp[t]$ to allocate the resource to (i.e., grant the trip request).
% Once the resource is allocated, a new instance of the resource competition is formed with a new pair of randomly matched agents from the population $\Comp_{t_{n+1}} \subset \N$.

A karma mechanism works as follows.
Each agent $l \in \N$ in the population is endowed with a non-negative integer counter $k^l[t] \in \Int$, called \emph{karma}, which is private to the agent.
Moreover, an additional \emph{surplus karma} counter $\ksur[t] \in \Int$ exists in the system.

At each interaction time $t$, each agent $i \in \Comp[t]$ involved in the resource competition submits a sealed non-negative integer bid $b^i[t] \in \{0,\dots,k^i[t]\}$, which is bounded by the agent's karma. The outcome of the interaction is determined by a \emph{resource allocation rule} and a \emph{payment} rule.
The \emph{resource allocation rule} decides which of the two competing agents is selected to receive the contended resource.
%(i.e., pass first in the intersection). 

\smallskip
\begin{center}
    \fbox{
        \begin{minipage}{0.92\textwidth}
            \begin{center}
                \textbf{Resource allocation rule}
            \end{center}
            
            \parbox{15mm}{\textbf{Input:}} 
            sealed bids $b^i[t]$ of the two competing agents $i \in \Comp[t]$.
            
            \medskip
            
            \parbox{15mm}{\textbf{Output:}} 
            selected agent $i^*[t] \in \Comp[t]$ to receive the resource,
            
            \[
                i^*[t] = \argmax_{i \in \Comp[t]} \: b^i[t].
            \]
        \end{minipage}
    }
\end{center}
\smallskip

It is natural to consider a resource allocation rule that allocates the resource to the agent with the highest bid.
A tie breaking rule is needed when both bids coincide, and we use a fair coin toss for this purpose.

The \emph{karma payment rule} determines the karma payments of the two competing agents.

\def\paymentwinner{p^{i^*}[t]}
\def\paymentloser{p^{-i^*}[t]}

\smallskip
\begin{center}
    \fbox{
        \begin{minipage}{0.92\textwidth}
            \begin{center}
                \textbf{Karma payment rule}
            \end{center}
            
            \parbox{15mm}{\textbf{Input:}}  \begin{minipage}[t]{\textwidth}
                sealed bids $b^i[t]$ of the two competing agents $i \in \Comp[t]$;\\
                selected agent $i^*[t]$ and yielding agent $-i^*[t] = \Comp[t] \setminus i^*[t]$.
            \end{minipage}
            
            \medskip
            
            \parbox{15mm}{\textbf{Output:}}  \begin{minipage}[t]{\textwidth}
                %\begin{itemize}[leftmargin=*]
                    %\item Karma payment from the selected agent $\Deltaminus k_{t_n}^{i^*_{t_n}} \in \{0,\dots,b^{i^*_{t_n}}_{t_n}\}$.
                    karma payment of the selected agent 
                    $0 \le \paymentwinner \le b^{i^*}[t]$;\\
                    %\item Karma payment to the non-selected agent $\Deltaplus k_{t_n}^{-i^*_{t_n}} \in \{0,\dots,\Deltaminus k_{t_n}^{i^*_{t_n}}\}$.
                    karma payment of the yielding agent 
                    $-\paymentwinner \le \paymentloser \le 0$.
                    %\item Karma payment to the central ledger $\Deltaplus K_{t_n} = \Deltaminus k_{t_n}^{i^*_{t_n}} - \Deltaplus k_{t_n}^{-i^*_{t_n}}$.
                    %\item surplus karma $\Deltaplus k_{t_n}^\text{s} = \Deltaminus k_{t_n}^{i^*_{t_n}} + \Deltaminus k_{t_n}^{-i^*_{t_n}}$.
                %\end{itemize}
            \end{minipage}
        \end{minipage}
    }
\end{center}
\smallskip

Note that the yielding agent makes a non-positive payment (i.e., it receives karma). As a consequence of this payment rule, at each interaction time $t$, the karma counters are updated as follows:
\begin{align*}
    k^{i^*}[t+1] & \leftarrow k^{i^*}[t] - \paymentwinner
    && \text{(selected agent $i^*[t]$)}, \\
    k^{-i^*}[t+1] & \leftarrow k^{-i^*}[t] - \paymentloser
    && \text{(yielding agent $-i^*[t]$)}, \\
    \ksur[t+1] & \leftarrow \ksur[t] + \paymentwinner + \paymentloser
    && \text{(surplus karma)}.
\end{align*}
Examples of karma payment rules are presented in Section~\ref{subsec:KarmaTransfer}. The surplus karma $\ksur[t]$ is meant to keep track of any excess karma payment in the interaction, in case $\paymentwinner + \paymentloser \neq 0$, such that this excess gets \emph{redistributed} to the population agents.
This redistribution occurs at time instants $\tredist$ in accordance with a \emph{karma redistribution rule}.

\def\redistribution{r^l[\tredist]}
%\Deltaplus k^i_{t^r_m}

\smallskip
\begin{center}
    \fbox{
        \begin{minipage}{0.92\textwidth}
            \begin{center}
                \textbf{Karma redistribution rule}
            \end{center}
            
            \parbox{15mm}{\textbf{Input:}}  \begin{minipage}[t]{0.7\textwidth}
                %\begin{itemize}[leftmargin=*]
                    %\item Discrete time instances $\{t^r_m\}_{m \in \Int}$ at which to perform the redistribution.
                    karma $k^l[\tredist]$ of all the agents $l \in \N$;\\
                    surplus karma $\ksur[\tredist]$.
                %\end{itemize}
            \end{minipage}
            
            \medskip
            
            \parbox{15mm}{\textbf{Output:}}  \begin{minipage}[t]{0.7\textwidth}
                %\begin{itemize}[leftmargin=*]
                    %\item Karma payment from the central ledger $\Deltaminus K_{t^r_m} \in \{0,\dots,K_{t^r_m}\}$.
                    %\item Karma payments to each agent $\Deltaplus k^l_{t^r_m}$ satisfying $\sum_l \Deltaplus k^l_{t^r_m} = \Deltaminus K_{t^r_m}$.
                    karma redistribution $\redistribution$
                    to each agent $l \in \N$.
%                    \item surplus karma $\Deltaminus k_{t^r_m}^\text{s} = \Deltaminus k_{t_n}^{i^*_{t_n}} + \Deltaminus k_{t_n}^{-i^*_{t_n}}$.
 %                                       $\Deltaminus K_{t^r_m} \in \{0,\dots,K_{t^r_m}\}$.
%                \end{itemize}
            \end{minipage}
        \end{minipage}
    }
\end{center}
\smallskip

As a consequence of this redistribution rule, at each redistribution time $\tredist$ the karma counters are updated as follows:
\begin{align*}
    k^{l}[\tredist+1] & \leftarrow k^{l}[\tredist] + \redistribution
    && \text{(every agent $l$)}, \\
    \ksur[\tredist+1] & \leftarrow \ksur[\tredist] - \sum_l \redistribution 
    && \text{(surplus karma)}.
\end{align*}

% \begin{center}
%     \fbox{
%         \parbox{0.9\textwidth}{
%             \begin{center}
%                 \textbf{Karma transfer rule}
%             \end{center}
            
%             \textbf{Input:\phantom{Output}} \begin{minipage}[t]{0.7\textwidth}
%                 \begin{itemize}[leftmargin=*]
%                     \item Sealed bids $b^i$ of the involved agents $i \in \Comp_{t_n}$.
%                     \item Selected agent $i^*_{t_n}$ and non-selected agent $-i^*_{t_n} = \Comp_{t_n} \setminus i^*_{t_n}$.
%                 \end{itemize}
%             \end{minipage}
            
%             \bigskip
            
%             \textbf{Output:\phantom{Input}} \begin{minipage}[t]{0.7\textwidth}
%                 \begin{itemize}[leftmargin=*]
%                     \item Change of karma of the selected agent $-b^{i^*_{t_n}}_t \leq \Delta k^{i_t^*}_t \leq 0$.
%                     \item Change of karma of the non-selected agent $0 \leq \Delta k^{-i_t^*}_t \leq - \Delta k^{i_t^*}_t$.
%                     \item Excess karma $\delta k_t = -\Delta k^{i_t^*}_t - \Delta k^{-i_t^*}_t$.
%                 \end{itemize}
%             \end{minipage}
            
%             \medskip
%         }
%     }
% \end{center}

A considerable freedom in the design of the karma payment and redistribution rules is possible.
Hereafter, we present some possible examples, and in Section~\ref{sec:NumericalAnalysis} we demonstrate how the choice of these rules affects the strategic behavior of the agents and allows the system designer to achieve different resource allocation objectives.

\subsection{Examples of karma payment rules}
\label{subsec:KarmaTransfer}

\smallskip
\begin{center}
    \fbox{
        \begin{minipage}{0.92\textwidth}
            \begin{center}
                \textbf{Pay bid to peer} ($\texttt{PBP}$)
            \end{center}
            
            The selected agent pays its bid directly to the yielding agent, i.e.,
            \[
                \paymentwinner = -\paymentloser = b^{i^*}[t].
            \]
        \end{minipage}
    }
\end{center}
\smallskip

$\texttt{PBP}$ is an example of completely peer-to-peer karma payment rules. It has the advantage of not requiring the system-level surplus karma counter $\ksur [t]$ or any system-wide karma redistribution.
% In our motivating example, the two competing AVs can perform the transaction in a completely decentralized manner, e.g., through a cryptographic implementation.

% \smallskip
% \begin{center}
%     \fbox{
%         \parbox{0.9\textwidth}{
%             \begin{center}
%                 \textbf{Pay one to peer} ($\texttt{P1P}$)
%             \end{center}
            
%             The selected agent pays $1$ directly to the yielding agent (if it bid $1$ or more), i.e.,
%             \[
%                 \paymentwinner = - \paymentloser = \min\{1,b^{i^*_{t_n}}_{t_n}\}.
%             \]
%         }
%     }
% \end{center}

\smallskip
\begin{center}
    \fbox{
        \begin{minipage}{0.92\textwidth}
            \begin{center}
                \textbf{Pay bid to society} ($\texttt{PBS}$)
            \end{center}
            
            The selected agent pays its bid and therefore creates surplus karma (to be later redistributed), i.e., 
            \[
                \paymentwinner = b^{i^*}[t], \quad 
                \paymentloser = 0.
            \]
        \end{minipage}
    }
\end{center}
\smallskip

In contrast to $\texttt{PBP}$, $\texttt{PBS}$ is an example of a karma payment rule in which surplus karma is generated and needs to be redistributed.
We will demonstrate in the numerical analysis in Section~\ref{sec:NumericalAnalysis} that such a redistributive scheme can lead to higher levels of efficiency and fairness of the resource allocation.

% \smallskip
% \begin{center}
%     \fbox{
%         \parbox{0.9\textwidth}{
%             \begin{center}
%                 \textbf{Pay bid to peer with fee} ($\texttt{PBP+FEE}$)
%             \end{center}
            
%             Let $0 \leq v[b] \leq b$ be a non-decreasing fee based on the bid $b$.
%             The selected agent pays its bid minus the fee directly to the yielding agent, while the fee contributes to the surplus karma (to be later redistributed), i.e.,
%             \[
%                 \paymentwinner = b^{i^*_{t_n}}_{t_n}, \quad
%                 -\paymentloser = b^{i^*_{t_n}}_{t_n} - v[b^{i^*_{t_n}}_{t_n}].
%             \]
%         }
%     }
% \end{center}
% \smallskip

% Note that both $\texttt{PBP}$ and $\texttt{PBS}$ are special cases of $\texttt{PBP+FEE}$, the former with $v[b] = 0$ and the latter with $v[b]=b$.

\subsection{Examples of karma redistribution rules}

There are a plethora of methods to redistribute the surplus karma to the agents in the population, as the designer has the freedom to decide both the redistribution times $\tredist$ and the redistribution rule. 
Nevertheless, we assume that redistribution rules are intended to completely redistribute the surplus karma so that most of the karma in the system is held by the agents, i.e., the total karma held by the agents is a \emph{preserved quantity}. We will formalize this assumption in the analysis in Section~\ref{sec:KarmaTransferModel}, where we will effectively assume that the surplus is kept at zero (either by not generating surplus or by redistributing it immediately).

One possibility is for the time instants $\tredist$ to be periodic events in which the redistribution occurs (e.g., every day at midnight).
If $\ksur[\tredist]$ is not an integer multiple of the number of agents $N$, %either the whole amount can be non-uniformly redistributed such that each agent receives an integer amount, or 
then a remainder is left to be redistributed in the next period.
Another possibility is for the redistribution to occur asynchronously whenever $\ksur[t]$ exceeds $N$, so that a unit of karma per agent can be transferred. 
Finally, a non-uniform (possibly randomized) redistribution is possible.

% \smallskip
% \begin{center}
%     \fbox{
%         \parbox{0.9\textwidth}{
%             \begin{center}
%                 \textbf{Uniform karma redistribution} (for $\texttt{PBS}$, $\texttt{PBP+FEE}$, ...)
%             \end{center}
%             Surplus karma is transfered equally to all agents, i.e. 
%             \begin{align*}
%                 \redistribution = \left\lfloor \frac{\ksur_{\tredist_n}}{N}\right\rfloor.
%             \end{align*}
%         }
%     }
% \end{center}
% \smallskip

Beyond these examples, we will not detail here the specifics of the karma redistribution rule.
%
%In Section~\ref{subsec:KarmaTransferModel} we will consider a specific synchronous karma redistribution rule with further assumptions on the times at which the agents interact $\{t_n\}_{n \in \Int}$ for the purpose of a simpler analysis.
%
However, it is interesting to notice that redistributions need not to be limited to positive karma values. For example, it is possible to reduce the karma of each agent according to a non-decreasing ``tax'' function $0 \leq h[k] \leq k$, in order to create surplus (which will then be redistributed). 
We will briefly discuss the consequences of such a design in the numerical analysis in Section~\ref{sec:NumericalAnalysis}.

\section{A game-theoretic model for karma mechanisms}
\label{sec:Model}

In this section, we develop a game-theoretic model to facilitate the analysis of karma mechanisms. The goal of the model is to address the following points:
\begin{enumerate}
    \item Karma mechanisms induce a strategic scenario in which the agents must strategically choose their karma bids.
    The model will serve to demonstrate that this strategic scenario is well-posed by showing the existence of a suitable notion of equilibrium (the stationary Nash equilibrium, see Sections~\ref{subsec:Nash}--\ref{subsec:NashExists}).

    % \item Assuming that the agents are rational, a game-theoretic formulation will model the agents' strategic bidding behavior under the karma mechanism, which must be taken into account in the mechanism design.
    
    \item The model enables a computational tool to compute stationary best-response behavior of the agents. This tool can be used by the agents to derive their optimal bidding strategies.
    
    \item The specifics of the karma mechanism (e.g., the karma payment and redistribution rules) affect the strategic behavior of the agents, which in turn affects population-level design objectives, in non-trivial ways.
    The model will serve as an important tool to make mechanism design choices, as demonstrated in the numerical analysis in Section~\ref{sec:NumericalAnalysis}.
    
    % to predict social welfare measures under the karma mechanism and aid the system designer in making design choices such as the choice of the karma payment and redistribution rules.
    % What is the performance of the mechanism, in terms of the social welfare measures introduced? Performance can be assessed in large-scale agent-based simulations, but an analytical expression of the measures will help predict the performance without running costly simulations.
\end{enumerate}

The game played under the karma mechanism has a number of complicating features: it is an infinite dynamic game involving a large number of anonymous agents who have private states which depend on their past actions and, in turn, affect their future available actions.
For this purpose, we build our model on the class of \emph{dynamic population games}, following the formalism of~\cite{elokda2021dynamic}.

\subsection{The karma dynamic population game}
\label{subsec:KarmaDPG}

\subsubsection{Population model}

We consider that the number of agents $N$ is large such that they approximately form a \emph{continuum of mass}.
This is reasonable to consider for many of our envisioned applications, e.g., the number of ride-hailing riders in a metropolitan area is typically large.
We take the point of view of an \emph{ego agent} playing \emph{against the population} from which a random anonymous opponent is uniformly drawn in every resource competition instance.
% For the sake of the simplicity of the analysis, we assume that this ra that the opponent is uniformly selected from the population.
Let $i$ be the identity of the ego agent, and $t^i$ be the time instants at which the ego agent is involved in the resource competition, i.e., $i \in \Comp[t^i]$ for all $t^i$.
These are the only time instants of relevance to the ego agent, and therefore we model its state dynamics as a discrete-time Markov chain, with the discrete update events occurring at the global times $t^i$.
% These are the only relevant time instances to the strategic situation that the ego agent faces, which can be described as a Markov decision process
% with respect to the discrete-time Markov chain with time index $s \in \Int$. 
To simplify notation, we will drop the explicit time dependency since the only time instances of interest are a \emph{current time} of the ego agent $t^i$ and a \emph{next time} of the ego agent $t^{i+} = \min \{t > t^i \mid i \in \Comp[t]\}$.
We will also drop the superscript $i$ since all quantities belong to the ego agent, unless explicitly stated otherwise.
For example, we write $k$ instead of $k^i[t^i]$ to denote the ego agent's current karma, and $k^+$ instead of $k^i[t^{i+}]$ to denote its next karma.

\subsubsection{Agents' type and state}

Each ego agent has a \emph{private static type} $\tau \in \Types = \{\tau_1,\dots,\tau_{n_\tau}\}$.
The distribution of the agents' types in the population is specified by the parameter $g \in \Delta(\Types)$, with $g_\tau$ denoting the fraction of agents belonging to type $\tau$.

Moreover, each ego agent has a \emph{private time-varying state} $x$ which consists of an \emph{urgency state} $u$ and the \emph{karma} $k$ of the agent, i.e.,
\begin{align*}
    x = [u, k] &\in \X = \U \times \Int, & u &\in \U = \{u_1,\dots,u_{n_u}\}, & k &\in \Int.
\end{align*}

The urgency state represents a private valuation for the resource and takes one of the values in the discrete and finite set $\U$.
It therefore corresponds to the cost incurred by the agent when they cannot procure the resource.
For example, each value in $\U$ could correspond to different classes of trips and how important it is that the agent secures a ride-hail for those trips.
%For example, $\U=\{1,10\}$ denotes that the agent can be either lowly urgent ($u=1$) or highly urgent ($u=10$).
The urgency at consecutive resource competition instances of an ego agent of type $\tau$ follows an exogenous, irreducible Markov chain process with transition probabilities denoted by
\begin{align}
    \label{eq:urgency-process}
    \phi_\tau[u^+ \mid u].
\end{align}
This process allows to model different assumptions on the temporal preferences of the agent.
For example, a static urgency process models that the agent has an equal need for the resource at all times.
Alternatively, the process can encode that the agent experiences exogenous events of high urgency where its inconvenience for failing to acquire the resource is elevated
(e.g., the cost of failing to secure a ride-hailing trip as a function of the trip length or purpose).
% (e.g., the cost of waiting at an intersection as a function of the trip purpose).
While it is possible to extend the analysis to the case where the agent's next urgency is affected by the local outcome of the interaction (e.g., failing to secure a ride-hailing trip today leads to higher urgency tomorrow due to accumulated delay), we insist on considering the urgency to be a function of exogenous events (e.g., whether the agent must catch a flight today).
Moreover, we assume that the urgency processes of different agents are statistically independent.
% Moreover, we assume that the urgency processes of different agents in the population are statistically independent but homogeneous, i.e., they follow the same transition probabilities $\phi$.
% This assumption is not needed for the main technical results\footnote{One can readily incorporate heterogeneous agent types in the mathematical formulation via subscripts. For a set discrete types $\Types$, one can denote by $\phi_\tau[u^+ \mid u]$ the urgency transition probabilities specific to agents of type $\tau \in \Types$. All the forthcoming technical results extend immediately. Refer to~\cite{elokda2021dynamic} for additional discussion on dynamic population games with heterogeneous types. \label{ft:heter-urgency}}, and in practice we expect the agents' private urgency processes to be heterogeneous.
% We make this assumption in order to not overload the presentation, and because it facilitates the \emph{interpersonal comparability}~\cite{roberts1980interpersonal} of the rewards in our setting, which simplifies the definition of our social welfare measures (introduced in Section~\ref{subsec:SocialWelfare}).
% We will discuss the effects of heterogeneity in the urgency processes in the numerical analysis performed in Section~\ref{subsec:HeterogeneousUrgency}.

\subsubsection{Social state}

The joint distribution of the agents' types and states in the population is given by
\begin{align}
    \label{eq:StateDistSet}
    d \in \D = \left\{d \in \Real_+^{n_\tau \times \infty} \left\lvert \; \forall \tau \in \Types, \; \sum_{u,k} d_\tau[u,k] = g_\tau \right. \right\},
\end{align}
where $d_\tau[u,k]$ denotes the fraction of agents in type-state $[\tau, u,k]$.
% \footnote{In case of heterogeneous agent types, the distribution $d$ can be enriched to denote the joint type-state distribution, with $d_\tau[u,k]$ denoting the fraction of agents in a type $\tau$ and state $[u,k]$. \label{ft:heter-distribution}}.
The type-state distribution $d$ is a \emph{time-varying} quantity whose dynamics evolve in terms of the global times rather than the specific time instants of the ego agent.
However, we will be looking for conditions where this distribution is \emph{stationary} and therefore this difference is inconsequential.
%The initial state distribution is the parameter $d^0$, which implicitly determines the initial average amount of karma in the system
%\[
%    \kbar = \sum_{u,k} d^0[u,k] \: k.
%\]
%For example, an initial distribution with $\sum_u d^0[u,\kbar]=1$ models that all the agents are initially endowed with the same amount of karma $\kbar$.

The action of the ego agent is a non-negative integer bid which is limited by its karma
\[
    b \in \Bids^k := \{b \in \Int \mid b \leq k\}.
\]
The ego agent of type $\tau$ chooses its bid according to the \emph{homogeneous policy} of its type
\[
    \pi_\tau : \X \rightarrow \Delta(\Bids^k) = \left\{\sigma \in \Real_+^{k+1} \left\lvert \; \sum_b \sigma[b] = 1 \right. \right\},
\]
which maps its state $[u,k]$ to a probability distribution over the bids $b$.
% \footnote{In case of heterogeneous agent types, the policy of each type $\tau$ will be denoted by $\pi_\tau$. \label{ft:heter-policy}}.
We denote by $\pi_\tau[b \mid u,k]$ the probability of bidding $b$ when the agent of type $\tau$ is in state $[u,k]$.
The concatenation of the policies of all types $\pi=(\pi_{\tau_1},\dots,\pi_{\tau_{n_\tau}})$ is simply referred to as the policy.
The set of policies is denoted by $\Pi$.

The pair $(d,\pi) \in \D \times \Pi$ is referred to as the \emph{social state}\footnote{This terminology is adapted from~\cite{sandholm2010population}, where the social state refers to the distribution of actions played in the static population. Notice that we must also account for the distribution of states in our dynamic setting.}, as it gives a macroscopic description of the distribution of the agents' states, as well as how they behave. 

In order to characterize the Markov decision process that the ego agent faces, we will now turn to define an \emph{immediate reward function} $\reward_\tau[u,k,b](d,\pi)$ and a \emph{state transition function} $\transition_\tau[u^+,k^+ \mid u,k,b](d,\pi)$.
When the ego agent of type $\tau$ is in state $[u,k]$ in the current resource competition instance and it bids $b$, $\reward_\tau[u,k,b](d,\pi)$ gives its expected immediate reward, and $\transition_\tau[u^+,k^+ \mid u,k,b](d,\pi)$ gives the probability that, at its next resource competition, its state is $[u^+,k^+]$. 
Both the immediate reward and the state transition are functions of the social state $(d,\pi)$.

\subsubsection{Immediate reward function}

Since the ego agent gets matched with a random opponent from the population, an important quantity is the \emph{distribution of other agents' bids}, which can be readily derived from the social state as
\begin{align}
    \label{eq:Prob-bj}
    \bdist[b'](d,\pi) = \sum_{\tau', u', k'} d_\tau[u',k'] \: \pi_\tau[b' \mid u',k'],
\end{align}
where $b'$ (similarly, $\tau'$, $u'$, $k'$) denotes that these quantities belong to agents other than the ego agent.
On a fundamental level, the ego agent is playing a game against this distribution, since it determines the likelihood of being selected to receive the resource for a given bid $b$, as well as the likelihood of transitioning to the next karma $k^+$.
Let us denote the \emph{resource competition outcome} to the ego agent by $o \in \Oset = \{0,1\}$, where $o=0$ means that it is selected and $o=1$ that it is yielding.
Conditional on its bid $b$ and the opposing bid $b'$, the ego agent has the following probability of being selected
% \footnote{The probability of yielding is $\Prob[o=1 \mid b, b'] = 1 - \Prob[o=0 \mid b, b']$.}
\begin{align}
    \label{eq:ProbWinConditionBj}
    \Prob[o=0 \mid b, b'] = \begin{cases}
        1, &\text{if } b > b', \\
        0, &\text{if } b < b', \\
        0.5, &\text{if } b = b',
    \end{cases}
\end{align}
which lets us compute the probability of its resource competition outcome given its bid as a function of the social state as
\begin{align}
    \label{eq:ProbWin}
    \oprob[o \mid b](d,\pi) = \sum_{b'} \bdist[b'](d,\pi) \: \Prob[o \mid b, b'].
\end{align}

The ego agent incurs a cost equal to its urgency $u$ when it yields the resource $(o=1)$, and zero cost otherwise $(o=0)$. This allows us to define the immediate reward function as
\begin{align}
    \label{eq:KarmaRewards}
    \reward_\tau[u,k,b](d,\pi) = \reward[u,b](d,\pi) = -u \: \oprob[o=1 \mid b](d,\pi),
\end{align}
which is negated to denote a reward rather than a cost. Note that it only depends on the urgency $u$ and the bid $b$ and and not on the type $\tau$ or the karma $k$.
Note also that it is continuous in the social state $(d,\pi)$.

\subsubsection{State transition function}
\label{subsubsec:StateTransitions}

The urgency of the ego agent at its next resource competition instance follows the exogenous process $\phi_\tau[u^+ \mid u]$.
In contrast, the ego agent's next karma depends on multiple factors, including its current bid, the resource competition outcome, the specifics of the karma payment rule, and whether a karma redistribution event occurs before its next resource competition instance (see Figure~\ref{fig:timing}).
We abstract this dependency with the \emph{karma transition function}
\begin{align}
    \label{eq:KarmaTransitions}
    \kappa[k^+ \mid k,b,o](d,\pi),
\end{align}
which lets us express the state transition function as
\begin{align}
    \label{eq:KarmaStateTransitions}
    \transition_\tau[u^+,k^+ \mid u, k, b](d,\pi) = \phi_\tau[u^+ \mid u] \sum_o \oprob[o \mid b](d,\pi) \: \kappa[k^+ \mid k,b,o](d,\pi).
\end{align}
%which is a continuous function of the social state $(d,\pi)$.

\begin{figure}[bt!]
    \centering
    \begin{tikzpicture}[>=stealth']
        \draw[->, thick] (0,0) -- (10,0) node[right] {\scriptsize $t$};
        \draw (0.5,0.1) -- (0.5,-0.1) node[below] {\footnotesize $1$};
        \draw[blue, thick] (1.5,0.2) -- (1.5,-0.2) node[below] {\scriptsize $t^i_1=2$};
        \draw[blue, thick] (1.5,-0.2) -- (1.5,0.2) node[above] {\scriptsize $k^i[t^i_1]$};
        \draw (2.3,0.1) -- (2.3,-0.1) node[below] {\footnotesize $3$};
        \draw (2.8,0.1) -- (2.8,-0.1);
        \draw (3,0.1) -- (3,-0.1);
        \draw (3.5,0.1) -- (3.5,-0.1) node[below] {\footnotesize $\phantom{1}\cdots\phantom{1}$};
        \draw (3.5,-0.1) -- (3.5,0.1) node[above] {\scriptsize (other interactions)};
        \draw (4.5,0.1) -- (4.5,-0.1);
        \draw (5.1,0.1) -- (5.1,-0.1);
        \draw (5.2,0.1) -- (5.2,-0.1);
        \draw[blue, thick] (5.5,0.2) -- (5.5,-0.2) node[below] {\scriptsize $t^i_2$};
        \draw[blue, thick] (5.5,-0.2) -- (5.5,0.2) node[above] {\scriptsize $k^i[t^i_2]$};
        \draw (6.2,0.1) -- (6.2,-0.1);
        \draw (6.9,0.1) -- (6.9,-0.1);
        \draw[orange, thick] (7.25,0.2) -- (7.25,-0.2) node[below] {\scriptsize $\tredist$};
        \draw[orange, thick] (7.25,-0.2) -- (7.25,0.2) node[above] {\scriptsize redistribution};
        \draw (8,0.1) -- (8,-0.1);
        \draw (8.5,0.1) -- (8.5,-0.1);
        \draw[blue, thick] (9,0.2) -- (9,-0.2) node[below] {\scriptsize $t^i_3$};
        \draw[blue, thick] (9,-0.2) -- (9,0.2) node[above] {\scriptsize $k^i[t^i_3]$};
        \draw (9.5,0.1) -- (9.5,-0.1) node[below] {\footnotesize $\phantom{1}\cdots\phantom{1}$};
        
        \draw[->, blue] (1.75,0.8) to [bend left=20] node[midway, above] {\scriptsize $-p^i[t^i_1]$} (5.25,0.8);
        \draw[->, blue] (5.75,0.8) to [bend left=20] node[midway, above] {\scriptsize $-p^i[t^i_2] \color{orange} + r^i[\tredist]$} (8.75,0.8);
    \end{tikzpicture}

    \caption{Timeline of the resource competition instances, highlighting the times relevant for modelling an ego agent $i$'s karma Markov chain. A redistribution event could affect the karma transitions.}
    \label{fig:timing}
\end{figure}
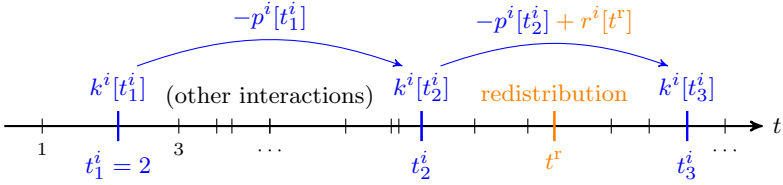

Section~\ref{sec:KarmaTransferModel} details the specifics of how to derive the karma transition function~\eqref{eq:KarmaTransitions} for the cases of \emph{pay bid to peer} ($\texttt{PBP}$) and \emph{pay bid to society} ($\texttt{PBS}$).
Here, we highlight two general properties of~\eqref{eq:KarmaTransitions}--\eqref{eq:KarmaStateTransitions} that are necessary for the main technical results to follow.

% We state two key assumptions that are necessary for the main technical results to follow.

\begin{assumption}[Continuity of the state transition function]
% \footnote{In case of heterogeneous agent types, continuity in $(d,\pi)$ must hold for the state transition function of each type. \label{ft:heter-continuity}}
\label{assumptions:continuity}
The state transition function $\transition_\tau[u^+,k^+ \mid u, k, b](d,\pi)$ defined in \eqref{eq:KarmaStateTransitions} is continuous in $(d,\pi)$.
\end{assumption}

The continuity assumption is barely restrictive since the dependency on $(d,\pi)$ typically arises in the form of expectations, as is demonstrated in Section~\ref{sec:KarmaTransferModel}.

\begin{assumption}[Karma preservation in expectation]
% \footnote{In case of heterogeneous agent types, the expectations in~\eqref{eq:KarmaPreservationDPGCompact}
% %, \eqref{eq:KarmaPreservationDPG} 
% must be taken with respect to the joint type-state distribution. \label{ft:heter-karma-preservation}}
\label{assumption:KP}
Karma is preserved for all $(d,\pi)$, when taking the expectation over the entire population, i.e., 
\begin{align}
    \label{eq:KarmaPreservationDPGCompact}
    &\E_{\substack{[\tau,u,k] \sim d \\b \sim \pi_\tau[\cdot \mid u,k]}}[k^+] = \E_{[\tau,u,k] \sim d}[k],
\end{align}
which expands to
\begin{multline}    
    \sum_{\tau,u,k} d_\tau[u,k] \sum_b \pi_\tau[b \mid u,k] \sum_o \oprob[o \mid b](d,\pi) \sum_{k^+} \kappa[k^+ \mid k,b,o](d,\pi) \: k^+ \\
    = \sum_{\tau,u,k} d_\tau[u,k] \: k. \tag{KP} \label{eq:KarmaPreservationDPG}
\end{multline}
\end{assumption}

Intuitively, Assumption~\ref{assumption:KP} requires that the karma held by the agents is preserved, either because no surplus karma is generated or because it is promptly redistributed.
This can be guaranteed by appropriate design of the payment rules or it can be achieved under some assumptions on the karma redistribution scheme, as is demonstrated in Section~\ref{sec:KarmaTransferModel}.

% Appendix~\ref{appendix:KarmaTransferModel} shows how to derive the karma transition function $\kappa$ for the cases of \emph{pay bid to peer} ($\texttt{PBP}$) and \emph{pay bid to society} ($\texttt{PBS}$), and how to verify that both Assumption~\ref{assumptions:continuity} and \ref{assumption:KP} hold.

% The right-hand side of~\eqref{eq:KarmaPreservationDPG} is the average amount of karma in the population when the state distribution is $d$.
% The left-hand side of~\eqref{eq:KarmaPreservationDPG} is the average amount of karma after the ego agent transitions assuming the role of all the agents in the population when the state distribution is $d$.
% The preservation of expected karma requirement~\eqref{eq:KarmaPreservationDPG} is similar to the karma preservation principle~\eqref{eq:KarmaPreservation}, but here we take the point of view of the stochastic evolution of the karma of the ego agent.
% It requires that the expected value of the ego agent's next karma equals the expected value of its current karma, when the state $[u,k]$ of the ego agent is sampled from the current state distribution $d$ and its bid is sampled from the conditional distribution dictated by the policy $\pi[\cdot \mid u,k]$.
% We will demonstrate in Section~\ref{subsec:KarmaTransferModel} how the karma transfer rules introduced in Section~\ref{subsec:KarmaTransfer} can be modelled in a manner that satisfies~\eqref{eq:KarmaPreservationDPG}.

\subsection{Solution concept: stationary Nash equilibrium}
\label{subsec:Nash}

\subsubsection{Best response}
We assume that the ego agent of type $\tau$ discounts its future rewards with the discount factor $\alpha_\tau \in [0,1)$.
% , which is the same for all the agents in the population\footnote{Similarly to the urgency process~\eqref{eq:urgency-process}, an extension to heterogeneous future discount factor types is immediate by keeping track of the types via subscripts (i.e., $\alpha_\tau \in [0,1)$ denotes the future discount factor of agents of type $\tau$). Refer to~\cite{elokda2021dynamic} for an additional discussion on dynamic population games with heterogeneous future discount factors. \label{ft:heter-discount-factor}}.
% We will test the robustness of the karma mechanisms against heterogeneity in the future discount factors in the numerical analysis performed in Section~\ref{subsec:HeterogeneousAlpha}.
Then, the expected immediate reward of the ego agent of type $\tau$ when it follows the policy $\pi_\tau$ is
\[
    R_\tau[u,k](d,\pi) = \sum_b \pi_\tau[b \mid u,k] \: \reward[u,b](d,\pi),
\]
and its state transition probabilities are
\[
    P_\tau[u^+,k^+ \mid u,k](d,\pi) = \sum_b \pi_\tau[b \mid u,k] \: \transition_\tau[u^+,k^+ \mid u,k,b](d,\pi).
\]
The \emph{expected infinite horizon reward} is therefore recursively defined as
\begin{multline}
    \label{eq:V-function-full}
    V_\tau[u,k](d,\pi) = R_\tau[u,k](d,\pi) \\
    + \alpha_\tau \sum_{u^+,k^+} P_\tau[u^+,k^+ \mid u,k](d,\pi) \: V_\tau[u^+,k^+](d,\pi).
\end{multline}
% or, equivalently in vector form,
% \begin{align}
%     \label{eq:V-function}
%     V(d,\pi) = (I - \alpha \: P(d,\pi))^{-1} \: R(d,\pi).
% \end{align}
Equation~\eqref{eq:V-function-full} is the well-known Bellman recursion for the fixed policy $\pi_\tau$. We next show that it has a unique solution that is continuous in $(d,\pi)$.
% , whose solution is unique and continuous in the social state\footnote{Written in vector form, $V(d,\pi)$ is the unique fixed point of the contraction map $T_{(d,\pi)}(v) = R(d,\pi) + \alpha P(d,\pi) v$, which is parametrized by $(d,\pi)$. It is straightforward to verify that the continuity of the fixed point $V(d,\pi)$ follows from the continuity of the map $T_{(d,\pi)}(v)$ in $(d,\pi)$. More precisely, one can show that for any two social states $(d,\pi)$ and $(d',\pi')$, the following bound holds $\|V(d,\pi) - V(d',\pi')\|_\infty \leq \frac{1}{1 - \alpha} \|T_{(d,\pi)}(V(d,\pi)) - T_{(d',\pi')}(V(d,\pi))\|_\infty$.}.

\begin{lemma}
Let Assumption~\ref{assumptions:continuity} hold. Then the solution of~\eqref{eq:V-function-full} is unique and continuous in $(d,\pi)$.
\end{lemma}
\begin{proof}
First we show uniqueness.
Let $V_\tau(d,\pi)$ be the vector formed by stacking $V_\tau[x](d,\pi)$ for all $x \in \X$. It is straightforward to show that $\|V_\tau(d,\pi)\|_\infty \leq \frac{u_\textup{max}}{1 - \alpha}$, where $u_\textup{max}$ is the maximal element of the finite set $\U$\footnote{This bound holds at a worst case where the ego agent yields with urgency $u_\textup{max}$ all the time.}.
Therefore, $V_\tau(d,\pi)$ lies in the Banach space of bounded sequences $(\ell^\infty, \|\cdot\|_\infty)$.
For a fixed social state $(d,\pi)$, let $T_\tau^{(d,\pi)} : \ell^\infty \rightarrow \ell^\infty$ be the map defined by the right-hand side of~\eqref{eq:V-function-full} (in vector form), i.e., $T_\tau^{(d,\pi)}(v) = R_\tau(d,\pi) + \alpha_\tau \: P_\tau(d,\pi) \: v$.
Observe that $V_\tau(d,\pi)$ is a fixed point of $T_\tau^{(d,\pi)}$, which we show to be unique by showing that $T_\tau^{(d,\pi)}$ is a contraction mapping, i.e.,
\begin{align}
    &\|T_\tau^{(d,\pi)}(v) - T_\tau^{(d,\pi)}(v')\|_{\infty} = \|\alpha_\tau \: P_\tau(d,\pi) \: (v - v')\|_{\infty} \notag \\
    &\quad= \alpha_\tau \: \max_x \left\lvert \sum_{x^+} P_\tau[x^+ \mid x](d,\pi) \: (v[x^+] - v'[x^+]) \right\rvert \notag \\
    &\quad \leq \alpha_\tau \: \max_x \sum_{x^+} \lvert P_\tau[x^+ \mid x](d,\pi) \: (v[x^+] - v'[x^+]) \rvert \notag \\
    &\quad \leq \alpha_\tau \left(\max_x \sum_{x^+} P_\tau[x^+ \mid x](d,\pi) \right) \left(\max_{x^+} \lvert v[x^+] - v'[x^+]\rvert \right) = \alpha_\tau \: \|v - v'\|_{\infty}. \label{eq:V-contraction}
\end{align}
% where the inequality holds since $P(d,\pi)$ is a stochastic matrix with largest eigenvalue 1.
Since $\alpha_\tau \in [0,1)$, this proves that $T_\tau^{(d,\pi)}$ is contractive.

Consider next the normed space $(\D \times \Pi, \|\cdot\|)$ (with an arbitrary norm) and the function $V_\tau : \D \times \Pi \rightarrow \ell^\infty$ defined as the unique fixed point of $T_\tau^{(d,\pi)}$.
We show that $V_\tau$ is continuous at every $(d,\pi) \in \D \times \Pi$.
Fix $\epsilon > 0$.
Choose $\epsilon' = (1 - \alpha_\tau) \: \epsilon$ and $\delta > 0$ such that $\|(d,\pi) - (d',\pi')\| < \delta \Rightarrow \|T_\tau^{(d,\pi)}(V_\tau(d,\pi)) - T_\tau^{(d',\pi')}(V_\tau(d,\pi))\|_\infty < \epsilon'$.
Such a $\delta$ is guaranteed to exist since $T_\tau^{(d,\pi)}(v)$ is continuous in $(d,\pi)$ for any fixed $v$ ($R_\tau(d,\pi)$ is continuous, and so is $P_\tau(d,\pi)$ as a consequence of Assumption~\ref{assumptions:continuity}).
Then, we have for any $(d',\pi')$ such that $\|(d,\pi) - (d',\pi')\| < \delta$,
\begin{align*}
    &\|V_\tau(d,\pi) - V_\tau(d',\pi')\|_\infty \\
    &\quad= \|V_\tau(d,\pi) - T_\tau^{(d',\pi')}(V(d,\pi)) + T_\tau^{(d',\pi')}(V_\tau(d,\pi)) - V_\tau(d',\pi')\|_\infty \\
    &\quad\leq \|T_\tau^{(d,\pi)}(V_\tau(d,\pi)) - T_\tau^{(d',\pi')}(V_\tau(d,\pi))\|_\infty \\
    &\quad\phantom{\leq}\quad+ \|T_\tau^{(d',\pi')}(V_\tau(d,\pi)) - T_\tau^{(d',\pi')}(V_\tau(d',\pi'))\|_\infty \\
    &\quad < (1 - \alpha_\tau) \: \epsilon + \alpha_\tau \: \|V_\tau(d,\pi) - V_\tau(d',\pi')\|_\infty,
\end{align*}
where the last inequality follows from~\eqref{eq:V-contraction}. Manipulating yields $\|V_\tau(d,\pi) - V_\tau(d',\pi')\|_\infty < \epsilon$, showing continuity.
\end{proof}

% Note that it is continuous in $(d,\pi)$, as $I - \alpha \: P(d,\pi)$ is guaranteed to be invertible for $\alpha \in [0,1)$\footnote{The stochastic matrix $P_\tau(d,\pi)$ has largest eigenvalue 1, therefore $I - \alpha_\tau \: P_\tau(d,\pi)$ cannot have a zero eigenvalue.}.
The ego agent's \emph{single-stage deviation reward} (commonly known as the \emph{Q-function}) is
\begin{multline}
    \label{eq:SingleStageDeviation}
    Q_\tau[u,k,b](d,\pi) = \reward[u,b](d,\pi) \\
    + \alpha_\tau \sum_{u^+,k^+} \transition_\tau[u^+,k^+ \mid u,k,b](d,\pi) \: V_\tau[u^+,k^+](d,\pi),
\end{multline}
which is the expected infinite horizon reward when the ego agent deviates from the policy $\pi_\tau$ for a single resource competition instance by bidding $b$ at the state $[u,k]$, then follows $\pi_\tau$ in the future resource competition instances.
Consequently, the \emph{state-dependent best response correspondence} of the ego agent is
\begin{multline}
    \label{eq:BestResponse}
    B_\tau[u,k](d,\pi) \\
    \in \left\{\sigma \in \Delta(\Bids^k) \left\lvert \; \forall \sigma' \in \Delta(\Bids^k), \; \sum\limits_b \left(\sigma[b] - \sigma'[b] \right) Q_\tau[u,k,b](d,\pi) \geq 0 \right. \right\}. \tag{BR}
\end{multline}
% \begin{align}
%     \label{eq:BestResponse}
%     B_{[u,k]}(d,\pi) \in \Set{\sigma \in \Delta(\Bids^k) | \begin{array}{l}
%          \forall \sigma' \in \Delta(\Bids^k), \\[2mm]
%          \quad \sum\limits_b \left(\sigma[b] - \sigma'[b] \right) Q[u,k,b](d,\pi) \geq 0
%     \end{array}}. \tag{BR}
% \end{align}
This is the set of probability distributions over the bids maximizing the expected single-stage deviation reward of the ego agent of type $\tau$ when its state is $[u,k]$ and the social state is $(d,\pi)$.

\subsubsection{Stationary Nash equilibrium}
We are now ready to define the solution concept that we adopt for this game.
\begin{definition}[Stationary Nash equilibrium]
\label{def:StationaryEquilibrium}
A stationary Nash equilibrium is a social state $(\sd,\epi) \in \D \times \Pi$ which satisfies for all $[\tau,u,k] \in \Types \times \U \times \Int$
\begin{align}
    \sd_\tau[u,k] &= \sum_{u^-,k^-} \sd_\tau[u^-,k^-] \: P_\tau[u,k \mid u^-,k^-](\sd,\epi), \label{eq:SNE-1} \tag{SNE.1} \\
    \epi_\tau[\cdot \mid u,k] &\in B_\tau[u,k](\sd,\epi). \label{eq:SNE-2} \tag{SNE.2}
\end{align}
\end{definition}
The stationary Nash equilibrium is similar to the classical notion of the Nash equilibrium in that it denotes a state of the game where agents have no incentive to unilaterally deviate from the equilibrium policies of their types $\epi_\tau$~\eqref{eq:SNE-2}, but additionally requires that the type-state distribution $\sd$ is \emph{stationary} under the stochastic processes characterized by the transition probabilities $P_\tau[u^+,k^+ \mid u,k](\sd,\epi)$~\eqref{eq:SNE-1}.
This stationarity condition implies that the ego agent need not consider the dynamics of the type-state distribution $\sd$ in its strategic behavior. Moreover, since the number of agents is large, the ego agent cannot unilaterally alter $\sd$ to further improve its rewards. Therefore, the equilibrium policies $\epi_\tau$ are indeed present and future optimal.

\subsection{Existence of stationary Nash equilibrium}
\label{subsec:NashExists}

In~\cite{elokda2021dynamic}, it is shown that a stationary Nash equilibrium is guaranteed to exist in every dynamic population game when the state space $\X$ is finite. 
We now extend this result to the karma dynamic population game, where the state space is countably infinite due to the karma state $k \in \Int$.
Observe that the set of type-state distributions $\D$, given in~\eqref{eq:StateDistSet}, is a convex subset of the Banach space of finitely summable infinite sequences $(\ell^1,\|\cdot\|_1)$.
This is because the elements $d \in \D$ can be represented as the infinite sequence $\{\sigma[n]\}_{n \in \Int}$ with
\begin{multline}
    \label{eq:Sequence}
    (\sigma[0],\dots,\sigma[n_\tau-1],\sigma[n_\tau],\dots,\sigma[n_\tau\:n_u-1],\sigma[n_\tau\:n_u],\dots) \\
    =(d_{\tau_1}[u_1,0],\dots,d_{\tau_{n_\tau}}[u_1,0],d_{\tau_1}[u_2,0],\dots,d_{\tau_{n_\tau}}[u_{n_u},0],d_{\tau_1}[u_1,1],\dots).
\end{multline}
Trivially, this sequence is finitely summable, with $\sum_n \lvert\sigma[n]\rvert = \sum_{\tau,u,k} d_\tau[u,k] = 1$.

Let us further restrict $\D$ to the subset of type-state distributions which respect a fixed average amount of karma $\kbar \in \Int$, denoted by:
\begin{align}
    \label{eq:Dkbar}
    \Dkbar = \left\{d \in \D \left\lvert \; \sum_{\tau, u, k} d_\tau[u,k] \: k= \kbar \right. \right\}.
\end{align}
This is also a convex subset of $\ell^1$. Furthermore, it is compact in $\ell^1$, as we show next
using the following auxiliary definition and lemma.

\begin{definition}[Equismall at infinity, \cite{treves1967topological}~p.451]
A subset $\Sigma$ of $\ell^1$ is said to be \emph{equismall at infinity} if, for every $\epsilon > 0$, there is an integer $n_\epsilon \geq 0$ such that
\[
    \sum_{n \geq n_\epsilon} \lvert\sigma[n]\rvert < \epsilon, \quad \text{for all } \sigma \in \Sigma.
\]
\end{definition}

\begin{lemma}[Compactness in $\ell^1$, {\cite[Theorem~44.2]{treves1967topological}}]
\label{th:Compactl1}
The following properties of a subset $\Sigma$ of $\ell^1$ are equivalent:
\begin{enumerate}[label=\alph*)]
    \item $\Sigma$ is compact;
    \item $\Sigma$ is bounded, closed, and equismall at infinity.
\end{enumerate}
\end{lemma}

\begin{corollary}
\label{cor:DkbarCompact}
$\Dkbar$ is a compact subset of $\ell^1$.
\end{corollary}

\begin{proof}
The set $\Dkbar$ is trivially closed since it is an intersection of closed polytopes.
It is also trivially bounded, since $0 \leq d_\tau[u,k] \leq 1$ for all $d \in \Dkbar$, $[\tau,u,k] \in \Types \times \U \times \Int$.
% It is also bounded under the 1-norm, since for any two elements $d$, $d' \in \Dkbar$:
% \[
%     \left|d - d'\right| = \sum_{\tau,u,k} \left|d_\tau[u,k] - d'_\tau[u,k]\right| \leq \sum_{\tau,u,k} \left(d_\tau[u,k] + d'_\tau[u,k]\right) = 2.
% \]
It therefore suffices to show that it is equismall at infinity.
For any $\epsilon > 0$, choose $k_\epsilon \in \Int$ such that $\frac{\kbar}{k_\epsilon} < \epsilon$, and $n_\epsilon = n_\tau \: n_u \: k_\epsilon$.
For an arbitrary $d \in \Dkbar$, let $\{\sigma[n]\}_{n \in \Int}$ be its sequence representation as given in~\eqref{eq:Sequence}.
We have
\begin{multline*}
    \kbar = \sum_{\tau,u,k} d_\tau[u,k] \: k \geq \sum_{\tau,u,k \geq k_\epsilon} d_\tau[u,k] \: k \geq k_\epsilon \sum_{\tau,u,k \geq k_\epsilon} d_\tau[u,k] = k_\epsilon \sum_{n \geq n_\epsilon} \lvert\sigma[n]\rvert \\
    \Leftrightarrow \sum_{n \geq n_\epsilon} \lvert\sigma[n]\rvert \leq \frac{\kbar}{k_\epsilon} < \epsilon.
\end{multline*}
\end{proof}

The compactness of $\Dkbar$ will enable us to invoke an infinite dimensional version of Kakutani's fixed point theorem to establish the existence of a stationary Nash equilibrium. Before we do, we need to ensure that the fixed point correspondence maps elements of $\Dkbar$ into itself. For a fixed policy $\pi \in \Pi$, define the map $W^\pi : \D \rightarrow \D$ as the concatination of the right-hand side of condition~\eqref{eq:SNE-1} for all $[\tau,u,k] \in \Types \times \U \times \Int$, i.e.,
\begin{align}
    W_{\tau}^\pi[u,k](d) = \sum_{u^-,k^-} d_\tau[u^-,k^-] \: P_\tau[u,k \mid u^-,k^-](d,\pi). \label{eq:W}
\end{align}
That $W^\pi$ maps elements of $\D$ to itself follows trivially from the fact that $P_\tau[u^+,k^+ \mid u,k](d,\pi)$ are transition probabilities.
We further have the following lemma.

\begin{lemma}
\label{lem:DkbarInvariant}
Let Assumption~\ref{assumption:KP} hold. Then for all $\kbar \in \Int$ and $\pi \in \Pi$, $W^\pi$ maps $\Dkbar$ into itself.
\end{lemma}

\begin{proof}
% It is helpful to write $W^\pi(d)$ explicitly
% \begin{align*}
%     W^\pi[u,k](d) &=\sum_{u^-,k^-} d[u^-,k^-] \sum_b \pi_\tau[b \mid u^-,k^-] \: \phi_{\mu_\tau}[u \mid u^-] \sum_o \oprob[o \mid b](d,\pi) \: \kappa[k \mid k^-,b,o](d,\pi).
% \end{align*}
For a $d \in \Dkbar$, the average amount of karma of $W_{\tau}^\pi(d)$ is
\begin{align*}
    &\sum_{\tau,u^+,k^+} W_{\tau}^\pi[u^+,k^+](d) \: k^+ \\
    &\quad= \sum_{\tau,u,k} d_\tau[u,k] \sum_{u^+,k^+} P_\tau[u^+,k^+ \mid u,k](d,\pi) \: k^+ \\
    &\quad= \sum_{\tau,u,k} d_\tau[u,k] \sum_b \pi_\tau[b \mid u, k] \sum_o \oprob[o \mid b](d,\pi) \sum_{k^+} \kappa[k^+ \mid k, b, o](d,\pi) \: k^+ \sum_{u^+} \phi_\tau[u^+ \mid u] \\
    &\quad= \sum_{\tau,u,k} d_\tau[u,k] \sum_b \pi_\tau[b \mid u, k] \sum_o \oprob[o \mid b](d,\pi) \sum_{k^+} \kappa[k^+ \mid k, b, o](d,\pi) \: k^+ \\
    &\quad= \sum_{\tau,u,k} d_\tau[u,k] \: k = \kbar,
\end{align*}
where we used the non-negativity of the summands to exchange the order of the infinite sums~\cite{rudin1976principles}, and condition~\eqref{eq:KarmaPreservationDPG}.
Therefore, $W^\pi(d) \in \Dkbar$.
\end{proof}

We are now ready to apply the following infinite dimensional fixed point theorem to establish our main technical result: the existence of a stationary Nash equilibrium in karma dynamic population games (Theorem~\ref{th:NashExists}).

\begin{lemma}[Kakutani-Glicksberg-Fan fixed point theorem, {\cite[Theorem~8.6]{granas2003fixed}}]
\label{th:Kakutani}
Let $C$ be a compact convex subset of a locally convex Hausdorff space $E$, and let $S : C \rightarrow 2^C$ be a set-valued correspondence which is upper hemicontinuous, nonempty, compact and convex. Then $S$ has a fixed point.
\end{lemma}

\begin{theorem}[Existence of a stationary Nash equilibrium in karma dynamic population games]
%Suppose that the karma transition function $\kappa[k^+ \mid k,b,o](d,\pi)$ satisfies preservation of expected karma~\eqref{eq:KarmaPreservationDPG}.
\label{th:NashExists}
Let Assumption~\ref{assumptions:continuity} and \ref{assumption:KP} hold.
Then for each $\kbar \in \Int$, a stationary Nash equilibrium $(\sd,\epi)$ satisfying $\sd \in \Dkbar$ is guaranteed to exist.
\end{theorem}

\begin{proof}
We can write the stationary Nash equilibrium conditions~\eqref{eq:SNE-1}--\eqref{eq:SNE-2} as the fixed points of the correspondence defined as
\[
    S(d,\pi) = \left(W^\pi(d), \{B_\tau[u,k](d,\pi)\}_{[\tau,u,k]}\right),
\]
where $\{B_\tau[u,k](d,\pi)\}_{[\tau,u,k]}$ is the sequence of best responses at all type-states $[\tau,u,k] \in \Types \times \U \times \Int$.
\begin{itemize}
    \item The set $C = \Dkbar \times \prod\limits_{\tau,u,k} \Delta(\Bids^k)$ is a compact subset of the locally convex Hausdorff space $E = \ell^1 \times \prod\limits_{\tau,u,k} \Real^{k+1}$, by Corollary~\ref{cor:DkbarCompact} and Tychonoff's theorem~\cite{tychonoff1930topologische}. $C$ is also trivially convex.
    \item $S$ maps $C$ into subsets of $C$, by Lemma~\ref{lem:DkbarInvariant} and the definition of the best response.
    \item $S$ is upper hemicontinuous and nonempty, by the continuity of $P_\tau(d,\pi)$ and $Q_\tau[u,k,b](d,\pi)$ in $(d,\pi)$, and Berge's maximum theorem~\cite{berge1997topological}.
    \item $S$ is compact and convex, since $W^\pi(d)$ is a singleton and $B_\tau[u,k](d,\pi)$ is the set of convex mixtures over the finite number of bids maximizing $Q_\tau[u,k,b](d,\pi)$.
\end{itemize}
It follows from Lemma~\ref{th:Kakutani} that $S$ is guaranteed to have a fixed point, which coincides with a stationary Nash equilibrium.
Since $\kbar$ above was arbitrary, this holds for each $\kbar \in \Int$.
\end{proof}

The significance of Theorem~\ref{th:NashExists} lies in establishing that karma mechanisms induce a well-posed game in which a rational behavior exists and is well defined.
Consequently, one can rigorously study the long-term social welfare implications of karma mechanisms at the stationary Nash equilibrium.
The uniqueness of the stationary Nash equilibrium, as well as whether different learning dynamics are guaranteed to converge to it, remain open research questions.

\subsection{Discussion of incentive compatibility}

We now turn to discuss the classical notion of \emph{incentive compatibility} (also known as \emph{strategy-proofness} or \emph{truthfulness}) in the context of karma mechanisms, in particular with respect to the bidding and resource allocation that happens at every interaction between the agents.
Following \cite{nisan2007algorithmic, krishna2009auction}, we say that an auction-like mechanism is incentive compatible if the optimal (selfish) action by each agent is to bid their own truthful evaluation of the contended resource, thus revealing their private preference.
Notice that, unlike classical monetary instruments, 
karma does not possess any intrinsic value, as it has no use outside of the game.
It does however acquire value as a means of exchange in the game, and one could attempt to define a notion of incentive compatiblity with respect to the value of karma given by the expected infinite horizon reward in~\eqref{eq:V-function-full}.
% We argue that such a notion of incentive compatibility is not useful for the following reasons.
% We argue that it does not apply in the classical sense for the following two reasons.
However we argue that such a notion of incentive compatibility is not critical for the efficiency of the resource allocation (i.e., the allocation of the resource to the highest urgency agent) at the stationary Nash equilibrium of the karma mechanism.
First, this is supported by the numerical analysis in Section~\ref{sec:NumericalAnalysis}, where near-optimal efficiency is robustly observed for a wide range of settings under all of the karma mechanism designs considered.
Second, unlike the classical monetary setting, incentive compatibility with respect to the value of karma does not guarantee efficiency.
This is because the value of karma depends on the contingent private state of the agent (the immediate urgency, but also the current karma balance and how urgent they expect to be in the future). For this reason, even if the karma mechanism was incentive compatible, a truthful bid would not be a perfect revelation of the agent's urgency.
It is important to highlight that in this work, we develop the tools to predict the strategic behavior under general karma mechanisms.
Therefore, we are able to robustly assess the resource allocation efficiency of these mechanisms when agents bid optimally according to their own self-interest, without the need for incentive compatibility as an intermediate step.

This is not to say that incentive compatibility is not a desirable property of the karma mechanism.
It will assist in the process of \emph{learning optimal policies} by providing optimal feedback when agents bid truthfully with respect to their value of karma (which also needs to be learnt; it corresponds to the value function that solves the Bellman equation given in~\eqref{eq:V-function-full}).
Moreover, it will likely lead to robustness against uncertain information about the social state.
The precise effect of the karma mechanism design on the learning process of the agents remains an exciting open research question.

\section{Modelling of karma payment and redistribution rules}
\label{sec:KarmaTransferModel}

We now revisit the karma payment and redistribution rules introduced in Section~\ref{sec:karma}, show how the karma transition function~\eqref{eq:KarmaTransitions} in Section \ref{subsubsec:StateTransitions} can be specialized to model them, and verify that they satisfy Assumption~\ref{assumptions:continuity} and \ref{assumption:KP}.

A key difference when it comes to deriving the karma transition function is whether there is no surplus karma generated by the payment rule, such as in \emph{pay bid to peer} ($\texttt{PBP}$), or whether surplus karma is generated, such as in \emph{pay bid to society} ($\texttt{PBS}$).
%, and when karma capital tax ($\texttt{TAX}$) is levied.
In the case of no surplus karma, the ego agent's karma at the next resource allocation instance is fully determined by the outcome of the current instance, making the karma transition probabilities easier to model.
In the case in which surplus is generated, then redistribution needs to occur, which in full generality might or might not happen between successive interactions of the ego agent (see Figure~\ref{fig:timing}). 
Extra care must be taken in order to guarantee that Assumption~\ref{assumption:KP} holds, and we will do so by introducing additional modelling assumptions that guarantee that the surplus karma is entirely redistributed between successive interactions of the ego agent.

\subsection{Payment rules with no surplus karma}
\label{subsec:PBP}

In the \emph{pay bid to peer} ($\texttt{PBP}$) karma payment rule,
% and remark that other payment rules with no surplus karma (e.g., $\texttt{P1P}$) can be modelled in a similar fashion.
the ego agent pays its bid if it is selected, and otherwise it gets paid the opposing bid $b'$. Consequently, the conditional probability of its next karma is
\[
    \Prob[k^+ \mid k,b,b',o] = \begin{cases}
        1, &\text{if } o = 0 \text{ and } k^+ = k - b, \\
        1, &\text{if } o = 1 \text{ and } k^+ = k + b', \\
        0, &\text{otherwise},
    \end{cases}
\]
which leads to the following karma transition function
\begin{align}
    \label{eq:kappa-PBP}
    \kappa^\texttt{PBP}[k^+ \mid k,b,o](d,\pi) = \cfrac{\sum\limits_{b'} \bdist[b'](d,\pi) \: \Prob[o \mid b, b'] \: \Prob[k^+ \mid k,b,b',o]}{\oprob[o \mid b](d,\pi)}.
\end{align}
It is straightforward to verify that~\eqref{eq:kappa-PBP} satisfies the continuity assumption (Assumption~\ref{assumptions:continuity}).
Karma preservation in expectation (Assumption~\ref{assumption:KP}) is also satisfied for all $(d,\pi)$ (that we omit from the notation), as
% \begin{align*}
%     &\sum_{u,k} d[u,k] \sum_b \pi[b \mid u,k] \sum_o \oprob[o \mid b](d,\pi) \sum_{k^+} \kappa^\texttt{PBP}[k^+ \mid k,b,o](d,\pi) \: k^+ \\
%     &\quad= \sum_{u,k} d[u,k] \sum_b \pi[b \mid u,k] \sum_{b'} \bdist[b'](d,\pi) \sum_o \Prob[o \mid b,b'] \sum_{k^+} \Prob[k^+ \mid k,b,b',o] \: k^+ \\
%     &\quad= \sum_{u,k} d[u,k] \sum_b \pi[b \mid u,k] \sum_{b'} \bdist[b'](d,\pi) \left(\Prob[o=0 \mid b,b'] \: (k - b) + \Prob[o=1 \mid b,b'] \: (k + b')\right) \\
%     &\quad = \sum_{u,k} d[u,k] \: k - \sum_{b',b>b'} \bdist[b] \: \bdist[b'] \: b + \sum_{b,b'>b} \bdist[b] \: \bdist[b'] \: b' 
%      = \sum_{u,k} d[u,k] \: k.
% \end{align*}
\begin{align*}
    &\sum_{\tau,u,k} d_\tau[u,k] \sum_b \pi_\tau[b \mid u,k] \sum_o \oprob[o \mid b] \sum_{k^+} \kappa^\texttt{PBP}[k^+ \mid k,b,o] \: k^+ \\
    &\quad= \sum_{\tau,u,k} d_\tau[u,k] \sum_b \pi_\tau[b \mid u,k] \sum_{b'} \bdist[b'] \sum_o \Prob[o \mid b,b'] \sum_{k^+} \Prob[k^+ \mid k,b,b',o] \: k^+ \\
    &\quad= \sum_{\tau,u,k} d_\tau[u,k] \sum_b \pi_\tau[b \mid u,k] \\
    &\qquad \quad \sum_{b'} \bdist[b'] \left(\Prob[o=0 \mid b,b'] \: (k - b) + \Prob[o=1 \mid b,b'] \: (k + b')\right) \\
    &\quad = \sum_{\tau,u,k} d_\tau[u,k] \: k - \sum_{b',b>b'} \bdist[b] \: \bdist[b'] \: b + \sum_{b,b'>b} \bdist[b] \: \bdist[b'] \: b' \\
    &\quad = \sum_{\tau,u,k} d_\tau[u,k] \: k.
\end{align*}

\subsection{Payment rules with surplus karma}
\label{subsec:PBS}

We make the following assumption to ease the modelling of payment rules which generate surplus karma to be redistributed to all the agents.
\begin{assumption}[Synchronous matching and redistribution]
\label{as:Synchronous}
At every time instant $t$, the whole population is randomly matched in simultaneous pairwise resource competition instances, and all surplus karma is redistributed immediately.
\end{assumption}

Under the \emph{pay bid to society} ($\texttt{PBS}$) karma payment rule, 
the ego agent pays its bid if it is selected, and pays nothing otherwise. Its conditional payment is hence given by
\[
    p^\texttt{PBS}[b,o] = \begin{cases}
        b, &\text{if } o = 0, \\
        0, &\text{otherwise}.
    \end{cases}
\]

Due to Assumption~\ref{as:Synchronous}, the average generated surplus can be computed by letting the ego agent assume the role of all the agents in the population, whose type-states are distributed as per $d$ and who follow the policies $\pi_\tau$ 
\begin{align*}
    \bar{p}^\texttt{PBS}(d,\pi) &= \sum_{\tau,u,k} d_\tau[u,k] \sum_b \pi_\tau[b \mid u,k] \sum_o \oprob[o \mid b](d,\pi) \: p^\texttt{PBS}[b,o] \\
    &=\sum_{\tau,u,k} d_\tau[u,k] \sum_b \pi_\tau[b \mid u,k] \: \oprob[o=0 \mid b](d,\pi) \: b.
\end{align*}
This gets redistributed to all the agents using the following integer-preserving redistribution rule (although other redistribution rules could be employed, as long as they redistribute the entire surplus):
%Let $\floor{\cdot}$ denote the truncation of $\cdot$ to the next-smallest integer, and $\ceil{\cdot}:=\floor{\cdot}+1$.
\begin{itemize}
    \item distribute $\floor{p^\texttt{PBS}(d,\pi)}$ to a fraction $f^\textup{low}(d,\pi):=\ceil{\bar{p}^\texttt{PBS}(d,\pi)} - \bar{p}^\texttt{PBS}(d,\pi)$ of agents, randomly selected;
    \item distribute $\ceil{\bar{p}^\texttt{PBS}(d,\pi)}$ to the remaining fraction $f^\textup{high}(d,\pi):= 1 - f^\textup{low}(d,\pi)$ of agents.
\end{itemize}
Consequently, the probability that the ego agent receives a surplus payment of $\floor{\bar{p}^\texttt{PBS}(d,\pi)}$ (respectively, $\ceil{\bar{p}^\texttt{PBS}(d,\pi)}$) is $f^\textup{low}(d,\pi)$ (respectively, $f^\textup{high}(d,\pi)$), resulting in the following karma transition function
\begin{multline}
    \label{eq:kappa-PBS}
    \kappa^\texttt{PBS}[k^+ \mid k,b,o](d,\pi) \\
    = \begin{cases}
        f^\textup{low}(d,\pi), &\text{if } o = 0 \text{ and } k^+ = k - b + \floor{\bar{p}^\texttt{PBS}(d,\pi)}, \\
        f^\textup{high}(d,\pi), &\text{if } o = 0 \text{ and } k^+ = k - b + \ceil{\bar{p}^\texttt{PBS}(d,\pi)}, \\
        f^\textup{low}(d,\pi), &\text{if } o = 1 \text{ and } k^+ = k + \floor{\bar{p}^\texttt{PBS}(d,\pi)}, \\
        f^\textup{high}(d,\pi), &\text{if } o = 1 \text{ and } k^+ = k + \ceil{\bar{p}^\texttt{PBS}(d,\pi)}, \\
        0, &\text{otherwise}.
    \end{cases}
\end{multline}
It is straightforward to verify that~\eqref{eq:kappa-PBS} satisfies the continuity assumption (Assumption~\ref{assumptions:continuity}).
Karma preservation in expectation (Assumption~\ref{assumption:KP}) is also satisfied, as
% \begin{align*}
%     &\sum_{u,k} d[u,k] \sum_b \pi[b \mid u,k] \sum_o \oprob[o \mid b](d,\pi) \: \sum_{k^+} \kappa^\texttt{PBS}[k^+ \mid k,b,o](d,\pi) \: k^+ \\
%     &\quad= \sum_{u,k} d[u,k] \sum_b \pi[b \mid u,k]  \\
%     &\qquad\qquad 
%     \left(\oprob[o=0 \mid b](d,\pi) \: (k - b + \bar{p}^\texttt{PBS}(d,\pi)) + \oprob[o=1 \mid b](d,\pi) \: (k + \bar{p}^\texttt{PBS}(d,\pi)) \right) \\
%     &\quad= \sum_{u,k} d[u,k] \: k + \bar{p}^\texttt{PBS}(d,\pi) - \sum_{u,k} d[u,k] \sum_b \pi[b \mid u,k] \: \oprob[o=0 \mid b](d,\pi) \: b 
%     = \sum_{u,k} d[u,k] \: k,
% \end{align*}
\begin{align*}
    &\sum_{\tau,u,k} d_\tau[u,k] \sum_b \pi_\tau[b \mid u,k] \sum_o \oprob[o \mid b] \: \sum_{k^+} \kappa^\texttt{PBS}[k^+ \mid k,b,o] \: k^+ \\
    &\quad= \sum_{\tau,u,k} d_\tau[u,k] \sum_b \pi_\tau[b \mid u,k]  \\
    &\qquad \quad \left(\oprob[o=0 \mid b] \: (k - b + \bar{p}^\texttt{PBS}) + \oprob[o=1 \mid b] \: (k + \bar{p}^\texttt{PBS}) \right) \\
    &\quad= \sum_{\tau,u,k} d_\tau[u,k] \: k + \bar{p}^\texttt{PBS} - \sum_{\tau,u,k} d_\tau[u,k] \sum_b \pi_\tau[b \mid u,k] \: \oprob[o=0 \mid b] \: b \\
    &\quad= \sum_{\tau,u,k} d_\tau[u,k] \: k,
\end{align*}
where we use $f^\textup{low}\: \floor{\bar{p}^\texttt{PBS}} + f^\textup{high} \: \ceil{\bar{p}^\texttt{PBS}} = \bar{p}^\texttt{PBS}$.

\section{Numerical analysis}
\label{sec:NumericalAnalysis}

In this section, we perform a numerical analysis of karma mechanisms, providing insights on the strategic behaviors that emerge at the stationary Nash equilibrium, and their consequences on the social welfare.
We first define the social welfare measures in Section~\ref{subsec:SocialWelfare}, then analyze the performance of the mechanisms in a demonstrative case study in Section~\ref{subsec:HighUrgency}.
Finally, we test the robustness of the mechanisms to heterogeneity of the agents in Sections~\ref{subsec:HeterogeneousAlpha} and \ref{subsec:HeterogeneousUrgency}.

As detailed in Appendix~\ref{sec:Computation}, all the stationary Nash equilibria presented were computed using a dynamic equilibrium-seeking algorithm that is inspired by \emph{evolutionary dynamics in population games}~\cite{sandholm2010population,elokda2021dynamic}.

\subsection{Social welfare measures and benchmark resource allocation schemes}
\label{subsec:SocialWelfare}

In order to quantitatively assess the performance of karma mechanisms, we introduce the following social welfare measures, along with benchmark resource allocation schemes that optimize them.
As a baseline, we take a resource allocation scheme that simply allocates the resource in every competition instance based on a fair coin toss. We denote this scheme by $\texttt{COIN}$.

\subsubsection{Efficiency}

We define \emph{efficiency} as
\begin{align}
    \label{eq:Efficiency}
    \text{eff} = \lim_{T \rightarrow \infty} \frac{1}{T} \E\left[\sum_{t = 0}^{T-1} \sum_{i \in \Comp[t]} \frac{\reward^i[t]}{2}\right],
\end{align}
which is the expected average reward of the two agents involved in the infinitely repeated resource competition instances.
At the stationary Nash equilibrium $(\sd,\epi)$ of the continuous population model, \eqref{eq:Efficiency} evaluates to 
\begin{align*}
    \label{eq:NashEfficiency}
    \text{eff}(\sd,\epi) = \sum_{\tau,u,k} \sd_\tau[u,k] \: R_\tau[u,k](\sd,\epi),
\end{align*}
which is the expected average reward per resource competition instance of an ego agent assuming the role of all the agents (leveraging the stationarity of $\sd$).
% In the subsequent agent-based simulations, we will approximate~\eqref{eq:Efficiency} empirically for a large but finite $T$.

A benchmark resource allocation scheme which maximizes the efficiency is known as the \emph{omniscient benevolent dictator}, who has access to the agents' private urgency and allocates the resource to the agent with highest one. We denote this scheme by $\CENTE$.

\subsubsection{Ex-post access fairness and ex-post reward fairness}

In line with the literature on randomized resource allocations (e.g., \cite{cappelen2013just}), the following \emph{ex-post fairness} measures are defined for finite time horizons $T$ and particular realizations of the repeated resource allocations\footnote{To the extent of our knowledge, ex-post fairness has not been defined in infinitely repeated settings thus far.
% One can also consider notions of \emph{ex-ante fairness}, i.e., fairness of the expected allocations, which are only interesting in our setting when there is heterogeneity in the agents. We briefly discuss ex-ante fairness notions in Sections~\ref{subsec:HeterogeneousAlpha} and \ref{subsec:HeterogeneousUrgency}.
}.

Let $w^i_T$ be the fraction of times agent $i$ was selected to receive the resource (with respect to the times it was involved in resource competitions).
The \emph{ex-post access fairness} is defined via the standard deviation of $w^i_T$ with respect to the different agents, i.e., 
\begin{align*}
    \accessfair_T &= -\std_{l \in \N} w^l_T, &
    w^i_T &= \frac{1}{T} \sum_{s=0}^{T-1} \left[i = i^*[t^i_s]\right]. 
\end{align*}
A benchmark dynamic resource allocation scheme which aims to maximize the ex-post access fairness is one that ensures that the agents \emph{take turns} accessing the resource, by selecting the agent who has received the resource the least fraction of times in the past. We denote this scheme by $\texttt{TURN}$.

Let instead $\bar{\reward}^i_T$ be an agent's mean reward. 
Then the \emph{ex-post reward fairness} is defined as the standard deviation of $\bar{\reward}^i_T$ with respect to the different agents, i.e., 
\begin{align*}
    \rewardfair_T &= -\std_{l \in \N} \bar{\reward}^l_T, &
    \bar{\reward}^i_T &= \frac{1}{T} \sum_{s=0}^{T-1} \zeta^i[t^i_s].
\end{align*}
Notice that, in contrast to the ex-post access fairness, ex-post reward fairness cannot be evaluated without knowing the private urgency of the agents.

\subsection{Case study: homogeneous agents with rare high-urgency state}
\label{subsec:HighUrgency}

We showcase our results in a scenario where the agents are homogeneous, i.e., they all follow the same urgency process $\phi$ and have the same future discount factor $\alpha$ (and there is only one type $\tau$, which we drop from the notation in this section).
We will investigate the role of heterogeneity in karma mechanisms in Sections~\ref{subsec:HeterogeneousAlpha}--\ref{subsec:HeterogeneousUrgency}. The agents are typically lowly urgent ($u=1$), and have a rare occurrence of being highly urgent ($u=10$).
The agents can anticipate when they will be highly urgent ahead of time.
This is represented by the following urgency process:\\
\begin{minipage}[c]{0.6\textwidth}
\vspace{2mm}
\begin{tikzpicture}[->,>=stealth',shorten >=1pt,auto,node distance=2.3cm,
                    semithick]
  \tikzstyle{every state}=[fill=white,draw=black,text=black,minimum height=12.6mm,align=center]

  \node[state] (A)    {\tiny default \\ \footnotesize $u=1$};
  \node[state] (B) [right of=A,yshift=6mm] {\tiny interm. \\ \footnotesize $u=1$};
  \node[state] (C) [right of=B,yshift=-6mm] {\tiny urgent \\ \footnotesize $u=10$};

  \path (A) edge [loop left]  node {\footnotesize 0.95} (A)
        (A) edge [bend left=10] node {\footnotesize 0.05} (B)
        (B) edge [in=130,out=160,loop] node [left] {\footnotesize 0.5} (B)
        (B) edge [bend right=10]  node [above] {\footnotesize 0.5} (C)
        (C) edge [bend left=10]  node {\footnotesize 0.95} (A)
        (C) edge [bend right]  node [above right] {\footnotesize 0.05} (B);
\end{tikzpicture}
\vspace{4mm}
\end{minipage}\hfill
\begin{minipage}[c]{0.3\textwidth}
\begin{equation}
    \label{eq:UrgencyProcess}
    \begin{aligned}
        \U &= \{1, 1, 10\}, \\
        \phi &= \begin{pmatrix}
            0.95 & 0.05 & 0 \\
            0 & 0.5 & 0.5 \\
            0.95 & 0.05 & 0
        \end{pmatrix}.
    \end{aligned}
\end{equation}
\end{minipage}\\
Notice that there are two low urgency states; the first is the `default' state in which the agents find themselves most of the times, and the second is an `intermediate' state which has a high probability of transitioning to the high urgency state.

For example, \emph{$u=1$ default} represents a regular day, \emph{$u=1$ intermediate} a regular day where the agent anticipates it must go to the airport during rush hour tomorrow, and \emph{$u=10$} the day of that important trip.

\begin{figure}[bt!]
    \centering
	\includegraphics[trim=27.95416pt 0 5.97237pt  0]{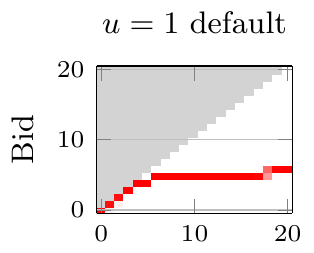}
	\hfil
	\includegraphics[trim=14.83821pt 0 13.0581pt  0]{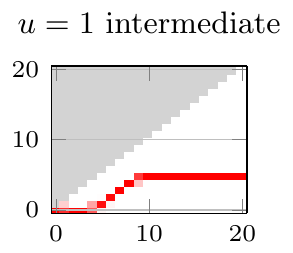}
	\hfil
	\includegraphics[trim=14.83821pt 0 5.97237pt  0]{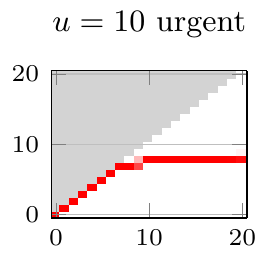}
	
	\includegraphics[trim=23.96802pt 0 10.86873pt  0]{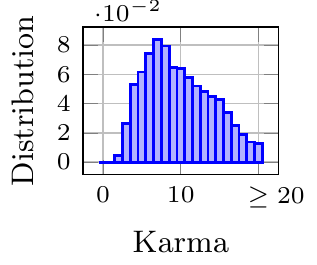}
	\hfil
	\includegraphics[trim=17.21321pt 0 10.86873pt  0]{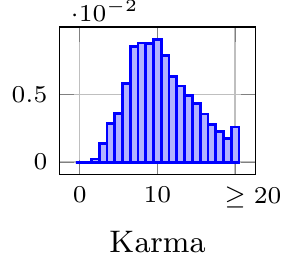}
	\hfil
	\includegraphics[trim=17.21321pt 0 10.86873pt  0]{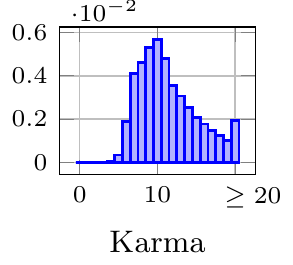}
	
	\caption{Stationary Nash equilibrium with urgency process \eqref{eq:UrgencyProcess}, karma rule $\texttt{PBS}$, future discount factor $\alpha=0.98$.}
% 	\caption{Stationary Nash equilibrium under $\texttt{PBS}$ for urgency process with rare high urgency event and future discount factor $\alpha=0.98$.}
	\label{fig:U-1-1-10-PBS-alpha-0.98}
\end{figure}

Figure~\ref{fig:U-1-1-10-PBS-alpha-0.98} shows the stationary Nash equilibrium computed
% using Algorithm~\ref{alg:StationaryEquilibrium} in Appendix~\ref{sec:Computation}}
for the case when the karma payment rule is \emph{pay bid to society} (\texttt{PBS}) and the agents discount their future rewards with $\alpha=0.98$. The average amount of karma per agent is $\kbar=10$.
The top of the figure shows the equilibrium bidding policy $\epi$ at each urgency state, where for a given level of karma ($x$-axis) the intensity of the red color denotes the probabilistic weight placed on the bids ($y$-axis), and disallowed bids that exceed the karma budget are displayed gray.
The bottom of the figure shows the stationary joint urgency-karma distribution $\sd$.
The stationary Nash equilibrium exhibits multiple intuitive behaviors.
First, agents bid parsimoniously, in order to save karma for the future rather than maximize their immediate chances of success.
Second, agents bid more in the high urgency state than in the low urgency states, thereby effectively signalling their urgency.
Interestingly, the agents bid zero when they are low on karma in the intermediate low urgency state, in order to gather karma for the anticipated high urgency state.
As a consequence, high urgency agents typically have more karma.

Figure~\ref{fig:U-1-1-10-PBP-PBS-performance} shows the performance of the karma mechanisms with respect to the social welfare measures of efficiency, ex-post access fairness and ex-post reward fairness, as a function of the agents' future discount factor $\alpha$.
In generating this figure, for each value of $\alpha$, we ran agent-based simulations with $N=200$ agents who were randomly matched in a total of $T=1000$ interactions per agent and bid according to the stationary Nash equilibrium policy for $\alpha$.
Each simulation was repeated 10 times in order to construct the displayed confidence intervals.
Efficiency and ex-post access fairness are plotted jointly as a trade-off chart on the left side of the figure, and the ex-post reward fairness is plotted on the right.
We compare the performance under the karma payment rules $\texttt{PBP}$ and $\texttt{PBS}$ and the benchmark resource allocation schemes
introduced in Section~\ref{subsec:SocialWelfare}.
% The cost efficiency metric is computed analytically as per~\eqref{eq:NashEfficiency} and overlaid with the range of values observed in the agent-based simulations.
% Additionally, the performance of the benchmark policies of the fair coin toss ($\texttt{COIN}$), as well as an all-knowing centralized policy that always prioritizes higher urgency agents ($\texttt{CENT}$), is shown.
As expected, the best efficiency is achieved by $\CENTE$, the best ex-post access fairness by $\texttt{TURN}$, and the baseline $\texttt{COIN}$ performs poorly in all measures.
Interestingly, the performance of $\texttt{PBP}$ coincides with $\texttt{COIN}$ when the agents are fully myopic ($\alpha=0$).
In this case, the equilibrium policy can be computed in closed form\footnote{Except for this special case, it is numerically difficult to compute $\epi$ for $\texttt{PBP}$ and low values of $\alpha$, and we start at $\alpha=0.15$.};
it is a dominant strategy for the agents to bid all their karma since there is no sense in saving it for the future.
Under $\texttt{PBP}$, this leads to all the karma in the system being in the possession of one single agent at a time, rendering an essentially random allocation among all the other agents who have no karma.
In contrast, this does not occur under $\texttt{PBS}$ due to the karma redistribution, and while the bidding all behavior is not efficient also under this payment rule, it preserves some of the turn-taking capability of the karma.
In fact, the performance of $\texttt{PBS}$ dominates that of $\texttt{PBP}$ across all values of $\alpha$ and for all of the social welfare measures considered, highlighting the advantage of incorporating a redistributive scheme rather than a strictly peer to peer scheme.
This advantage comes at a price, since redistributive schemes such as $\texttt{PBS}$ requires some degree of centralization to keep track of and redistribute the surplus karma.
In many cases, however, it is natural to consider that the agents have a reasonably high value of $\alpha$, since they are expected to remain in the system for long.
Interestingly, both $\texttt{PBP}$ and $\texttt{PBS}$ perform similarly well in these cases, exposing that the performance of the karma mechanisms is robust to the specifics of the mechanism design in many reasonable scenarios.
Remarkably, both payment rules approach the optimal efficiency of $\CENTE$, without ever having to access the agents' private urgency.
At the same time, they vastly outperform $\CENTE$ both in terms of ex-post access efficiency and ex-post reward fairness, demonstrating that the karma mechanisms are successful in both achieving fair turn-taking, as well as catering to the agents' varying temporal needs.
This occurs as long as the agents do have some future discounting.
A severe degradation in the ex-post fairness occurs in the ``pathological'' case when the agents do not discount their future ($\alpha=1$).
This interesting case is discussed separately in Appendix~\ref{sec:alpha1} as it requires different analysis tools to those presented in Section~\ref{sec:Model}.

\begin{figure}[bt!]
    \centering
	\includegraphics[valign=m, trim=38.32916pt 0 0 0]{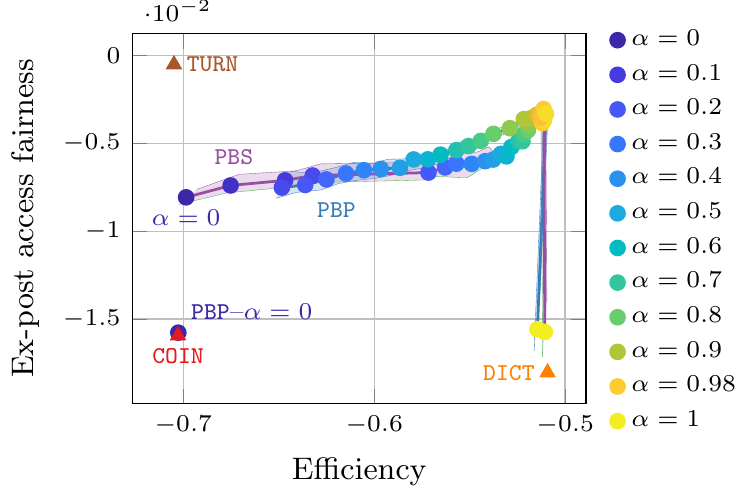}
	\hfil
	\includegraphics[valign=m, trim=0 0 0.2pt 0]{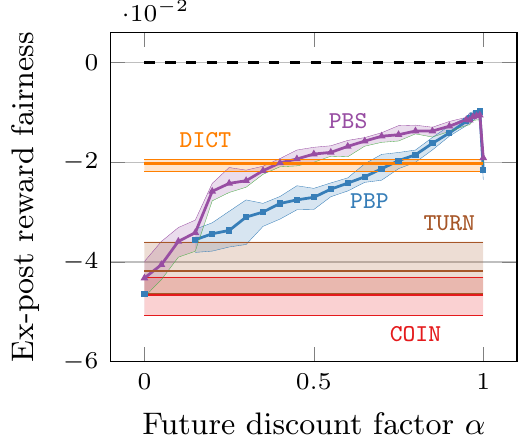}
	
	\caption{Performance of $\texttt{PBP}$ and $\texttt{PBS}$ karma payment rules when there is a rare high urgency event, as a function of the future discount factor $\alpha$.}
	\label{fig:U-1-1-10-PBP-PBS-performance}
\end{figure}

% Naturally, $\CENTE$ achieves the highest possible efficiency, but it is not expected to be particularly ex-post fair since it is memory-less.
% Indeed, $\CENTE$ performs poorly in terms of the ex-post access fairness, but decently well in terms of the ex-post reward fairness.
% This exposes a correlation between the efficiency and the ex-post reward fairness, attributed to the fact that if a highly urgent agent does not receive the resource, it incurs a cost equal to that of not receiving the resource 10 times when lowly urgent.
% The ex-post reward fairness is high when both the efficiency and ex-post access fairness are high.

\begin{figure}[bt!]
    \centering
    \includegraphics[trim=27.95416pt 0 2.60417pt 0]{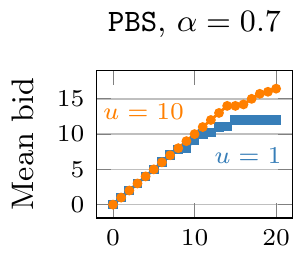}
	\hfil
	\vrule
	\hfil
    \includegraphics[trim=14.83821pt 0 2.60417pt 0]{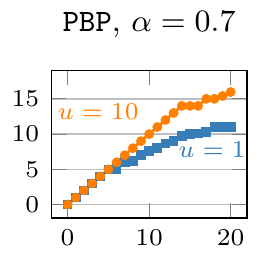}
	\hfil
	\includegraphics[trim=14.83821pt 0 2.60417pt 0]{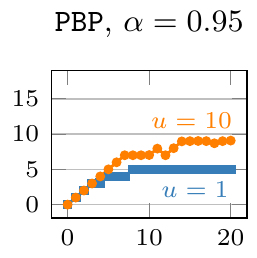}
	\hfil
	\includegraphics[trim=14.83821pt 0 2.60417pt 0]{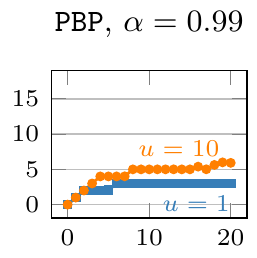}
	
	\includegraphics[trim=27.95416pt 0 10.86873pt 0]{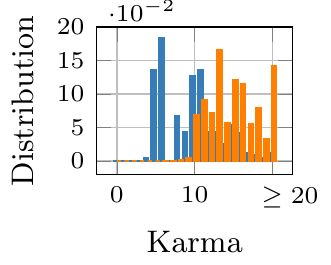}
	\hfil
	\vrule
	\hfil
	\includegraphics[trim=14.83821pt 0 10.86873pt 0]{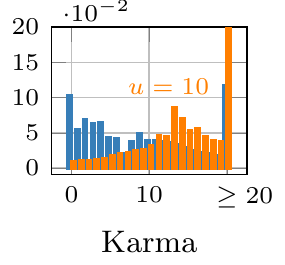}
	\hfil
	\includegraphics[trim=14.83821pt 0 10.86873pt 0]{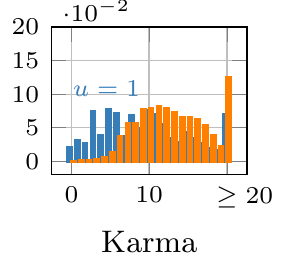}
	\hfil
	\includegraphics[trim=14.83821pt 0 10.86873pt 0]{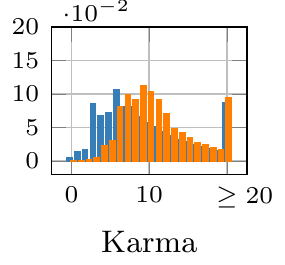}
	
	\caption{Comparison of stationary Nash equilibria under karma payment rules $\texttt{PBS}$ and $\texttt{PBP}$ for the low future discount factor $\alpha=0.7$ (left), and under $\texttt{PBP}$ for multiple future discount factors (right).}
	\label{fig:U-1-1-10-PBP-PBS}
\end{figure}

To provide insight into why $\texttt{PBS}$ outperforms $\texttt{PBP}$ for low values of the future discount factor $\alpha$, as well as why the performance of the karma mechanisms improves with increasing $\alpha$, we compare a number of stationary Nash equilibria in Figure~\ref{fig:U-1-1-10-PBP-PBS}.
Here, we compactly represent the equilibrium bidding policies through the \emph{mean bids},
% \footnote{For state $[u,k]$, the mean equilibrium bid is $\boldsymbol{\bar{b}}[u,k] = \sum_b \epi[b \mid u,k] \: b$. As Figure~\ref{fig:U-1-1-10-PBS-alpha-0.98} demonstrates, the equilibrium bidding policies are deterministic at most states, and therefore $\boldsymbol{\bar{b}}$ provides a good representation of $\epi$.}
and only present the results for the default low urgency and the high urgency states (omitting the intermediate low-urgency state).
The first two columns of the figure compare the stationary Nash equilibria under $\texttt{PBS}$ and $\texttt{PBP}$ for the relatively low value of $\alpha=0.7$.
Observe that under $\texttt{PBP}$, a significant mass of agents are expected to be low on karma when they are highly urgent, and therefore fail to signal their high urgency against a lowly urgent opponent, contributing to the loss of efficiency observed in Figure~\ref{fig:U-1-1-10-PBP-PBS-performance}.
In contrast, under $\texttt{PBS}$ the mass of the stationary karma distribution at the high urgency state is concentrated in a region where the agents will effectively outbid lowly urgent opponents most of the times, explaining the superior performance of $\texttt{PBS}$ in terms of efficiency.
This occurs due to the redistribution of karma, which ensures that the agents are sufficiently far from having critically low karma.

A similar mechanism is responsible for the improved efficiency at higher values of the future discount factor $\alpha$, as the rightmost three columns of Figure~\ref{fig:U-1-1-10-PBP-PBS} demonstrate by contrasting the stationary Nash equilibria under $\texttt{PBP}$ for $\alpha \in \{0.7,0.95,0.99\}$ (qualitatively similar results hold for $\texttt{PBS}$).
Instead of relying on karma redistribution, highly future aware agents learn to be sparing in the use of karma, in order to avoid the situation of being highly urgent and low on karma, in which karma loses its effectiveness as a signaling device.
This precisely exemplifies how repetition can be leveraged to align the agents' incentives, and suggests that karma is an effective instrument for this purpose.
Additionally, the mass of agents that have critically low karma values is generally much smaller for $\alpha=0.99$ than for $\alpha=0.7$ (at all urgency states), which contributes to the improved ex-post access fairness.

\subsection{Robustness to heterogeneous future discount factors}
\label{subsec:HeterogeneousAlpha}

\begin{figure}[bt!]
	\centering
	\includegraphics[trim=25.98193pt 0 10.00948pt 0]{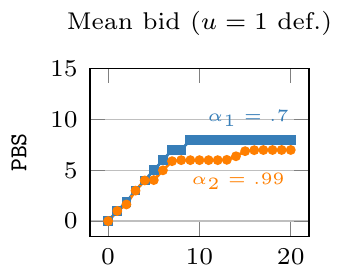}
	\hfil
	\includegraphics[trim=0 0 10.17004pt 0]{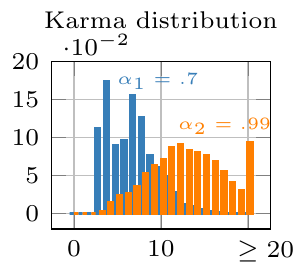}
	\hfil
	\includegraphics[trim=0 0 11.4978pt 0]{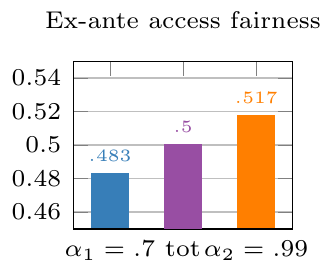}
	\hfil
	\includegraphics[trim=0 0 12.85405pt 0]{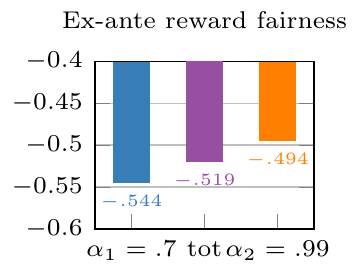}
	
	\includegraphics[trim=25.98193pt 23.02634pt 2.03955pt 0]{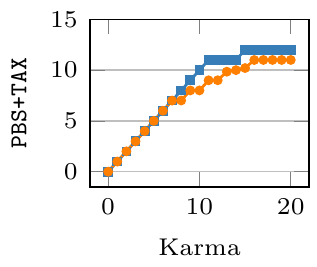}
	\hfil
	\includegraphics[trim=0 24.61179pt 10.17004pt 0]{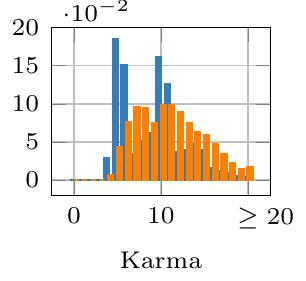}
	\hfil
	\includegraphics[trim=0 12.37706pt 7.84477pt 0]{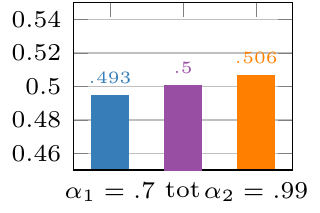}
	\hfil
	\includegraphics[trim=0 12.37706pt 7.84477pt 0]{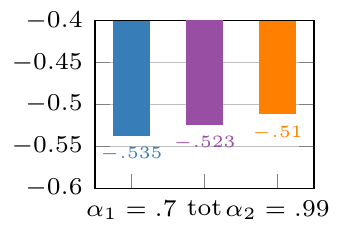}
	
	\vspace{24.61179pt}
	
	\caption{Robustness to heterogeneous future discount factors, without (top) and with (bottom) karma tax.}
	\label{fig:alpha1-0.7-alpha2-0.99}
\end{figure}

% This is the first of two sections which investigate the robustness of karma mechanisms to \emph{heterogeneity} in the population.
% All the results of Section~\ref{sec:Model} can be easily extended to account for populations with multiple types of agents\footnote{See footnotes~\ref{ft:heter-urgency}, \ref{ft:heter-distribution}, \ref{ft:heter-policy}, \ref{ft:heter-continuity}, \ref{ft:heter-karma-preservation}, \ref{ft:heter-discount-factor} and refer to~\cite{elokda2021dynamic} for the necessary technical steps.}.
In this section, we consider a mixed population where the agents can have one of two future awareness types; half of the agents heavily discount the future reward ($\alpha_1 = 0.7$) while the other half are strongly future-aware ($\alpha_2 = 0.99$).
This heterogeneity can have one of two interpretations. In the first interpretation, the future discount factors are true representatives of the agents' objectives, e.g., the $\alpha_1$ agents are expecting to exit the system sooner than the $\alpha_2$ agents.
In the second interpretation, all the agents are expected to remain in the system for the same (infinite) time, and the heterogeneity represents differences in their \emph{strategic competence}, i.e., the $\alpha_1$ agents are \emph{less patient} than the $\alpha_2$ agents.
This is the interpretation we focus on here.
We would like to investigate the extent to which karma mechanisms are gracious with respect to this difference.

Figure~\ref{fig:alpha1-0.7-alpha2-0.99} shows stationary Nash equilibrium results for $\texttt{PBS}$ (top) and $\texttt{PBS}+\texttt{TAX}$ (bottom),
where in the latter we collect a small progressive karma tax of the form $h[k] = 0.005 \: k^2$ from all the agents and redistribute it uniformly.
The mathematical modelling of the karma tax follows similar principles as the redistributive payment rule $\texttt{PBS}$ (see Section~\ref{subsec:PBS}).
The defining feature of Figure~\ref{fig:alpha1-0.7-alpha2-0.99} is that in the untaxed case, a slight `under-bidding' behavior of the $\alpha_2$ agents leads them to accumulate significantly more karma than the $\alpha_1$ agents on the long run.
As the \emph{ex-ante access fairness} and \emph{ex-ante reward fairness} plots demonstrate\footnote{When there are multiple agent types, we may have a systematic difference in the long-run probability of access to the resource as well as the mean rewards of the different types. This \emph{ex-ante} unfairness can be quantified directly from the stationary Nash equilibrium policies and distributions of the different types.},
this leads to some degree of unfairness between the two types, with the $\alpha_2$ agents getting access to the resource a higher fraction of times, as well as experiencing higher average rewards (although the disparity is reasonably small).
Nonetheless, applying a karma tax is an effective measure to equalize this disparity, since it both disincentivizes the $\alpha_2$ agents to hold on to too much karma, and also redistributes some of that karma to the $\alpha_1$ agents.
This serves as a demonstration of the freedom that karma mechanisms give to the system designer, who has a principled tool to achieve different resource allocation objectives.

\subsection{Robustness to heterogeneous urgency processes}
\label{subsec:HeterogeneousUrgency}

Thus far in our numerical analysis we have considered that all agents have the same urgency process.
% Although this assumption is not  necessary to derive the technical results of Section~\ref{sec:Model},
The homogeneity of the urgency process facilitates the \emph{interpersonal comparability}~\cite{roberts1980interpersonal} of utility between agents, and ultimately it allows to define a simple notion of efficiency.
It is important to notice, however, that each agent's urgency process is completely private and its only purpose is to encode the user temporal preference for when they would prefer to acquire the resource.
It enables to compare the value of the resource for the same agent at different times, more than enabling the comparison between different agents.

We have looked at the effect of introducing some \emph{invaders}\footnote{We borrow the term invaders from the standard nomenclature in evolutionary games.} in the population. We assumed that these invaders have a different urgency process in that they are in a high-urgency state more often than the nominal population. We have considered the case of a small and a large subpopulation of invaders.

The numerical results are reported in Table~\ref{tab:heterogeneous-urgency}.
It is evident that agents that present a higher frequency of the high-urgency state are not granted additional resources under the $\texttt{PBS}$ karma mechanism. 
On the contrary, the karma mechanism incentivizes agents to identify their most urgent instances parsimoniously.
In contrast, the benchmark strategy $\CENTE$ allocates additional resources to the high-urgency invading subpopulation.  This behavior illustrates what is a key feature of the proposed ``self-contained'' karma mechanism, that differentiates it from monetary schemes: fairness of the resource allocation emerges intrinsically from the mechanism and is not affected by exogenous factors of inequality between agents.
Uneven allocation of the resource is possible if desired, but it requires deliberate design choices such as non-uniform karma redistribution rules.

\begin{table}[tb!]
    \centering
    \caption{Probability of acquiring the resource in the case of agents with heterogeneous urgency processes.}
    \label{tab:heterogeneous-urgency}
    \small
    \begin{tabular}{@{}llcccc@{}}
        \toprule
         && \multicolumn{2}{c}{\texttt{PBS}} 
         & \multicolumn{2}{c}{\CENTE} \\ 
         \cmidrule(lr){3-4} \cmidrule(l){5-6}
         && \shortstack{nominal \\ agents}
         & invaders
         & \shortstack{nominal \\ agents}
         & invaders \\
         \midrule
         nominal & $10\%$ of time at $u=10$
         & 50\% & -- & 50\% & --
         \\
         \midrule
         \multirow{2}{*}{10\% invaders} & $20\%$ of time at $u=10$
         & 50.30\% & 47.34\% & 49.50\% & 54.50\%
         \\
         & $50\%$ of time at $u=10$
         & 50.71\% & 43.62\% & 48.00\% & 68.00\%
         \\
         \midrule
         \multirow{2}{*}{50\% invaders} & $20\%$ of time at $u=10$
         & 52.59\% & 47.41\% & 47.50\% & 52.50\%
         \\
         & $50\%$ of time at $u=10$
         & 53.76\% & 46.24\% & 40.00\% & 60.00\%
         \\
         \botrule
    \end{tabular}
\end{table}

% \begin{table}[tb]
%     \centering
%     \small
%     \begin{tabular}{@{}lp{25mm}ccccc@{}}
%          && nominal  
%          & \multicolumn{2}{c}{10\% invaders} 
%          & \multicolumn{2}{c}{50\% invaders} \\ 
%          \cmidrule(lr){3-3} \cmidrule(lr){4-5} \cmidrule(l){6-7}
%          &  & \shortstack{$10\%$ of time \\ at $u=10$}
%          & \shortstack{$20\%$ of time \\ at $u=10$}
%          & \shortstack{$50\%$ of time \\ at $u=10$}
%          & \shortstack{$20\%$ of time \\ at $u=10$}
%          & \shortstack{$50\%$ of time \\ at $u=10$}\\
%          \midrule
%          \multirow{2}{*}{\texttt{PBS}} & nominal agents 
%          & 50\% & 50.30\% & 50.71\% & 52.59\% & 53.76\% \\
%          & invaders
%          & -- & 47.34\% & 43.62\% & 47.41\% & 46.24\% \\
%          \midrule
%          \multirow{2}{*}{\CENTE} & nominal agents 
%          & 50\% & 49.50\% & 48.00\% & 47.50\% & 40.00\% \\
%          & invaders
%          & -- & 54.50\% & 68.00\% & 52.50\% & 60.00\% \\
%          \midrule \\
%     \end{tabular}
%     \caption{Probability of acquiring the resource in the case of agents with heterogeneous urgency processes.}
%     \label{tab:heterogeneous-urgency}
% \end{table}

% \newpage

\section{Conclusion}
\label{sec:Conclusion}

We have demonstrated the effectiveness of karma mechanisms for the dynamic allocation of common resources.
These mechanisms make it possible to achieve highly efficient and fair allocations when the resources are repeatedly disputed, without requiring access to the users' private preferences, and without resorting to monetary pricing, which is problematic in many important domains.
The efficiency and fairness of karma mechanisms is robustly observed in multiple numerical cases involving different mechanism designs, preference structures, and user heterogeneity.

We show that it is possible to rigorously study the strategic behavior of the users of a karma mechanism by modelling it as a dynamic population game, in which a stationary Nash equilibrium is guaranteed to exist.
We numerically investigate the karma stationary Nash equilibrium, providing insights on the strategic behaviors that emerge and on their consequences for the social welfare.
We also provide examples of how our model can be a versatile mechanism design tool for the system designer that wants to affect these behaviors and achieve different resource allocation objectives.

Future work includes applying karma mechanisms in the specific motivating use cases, which include the allocation of ride-hailing trips, autonomous intersection management, as well as traffic and/or internet congestion management.
We believe that many more applications are possible.
We would also like to investigate the surprisingly understudied notions of fairness in (infinitely) repeated resource allocations, and develop axiomatic principles and specifications to guide the design of karma mechanisms.
Moreover, our analysis suggests that karma mechanisms are robust to some types of user heterogeneity, but a comprehensive analysis of the practical effects of more forms of heterogeneity is desirable for some applications.
Finally, we remark that effective strategic play by the agents is how karma acquires value in a karma mechanism. An important open research question is how the users of a karma mechanism can learn their optimal bidding strategy from repeated play in a distributed fashion.

\section{Acknowledgement}
We would like to acknowledge anonymous reviewers whose comments and suggestions greatly improved the paper.
We are also thankful to Heinrich Nax for many fruitful discussions.
Research supported by NCCR Automation, a National Centre of
Competence in Research, funded by the Swiss National Science
Foundation (grant number $180545$).

\newpage

\begin{appendices}

\section{Computation of a stationary Nash equilibrium}
\label{sec:Computation}

Our stationary Nash equilibrium computation algorithm is motivated by the notion of \emph{evolutionary dynamics} in static population games. We write the equilibrium seeking problem as the problem of finding the rest points of the following continuous-time dynamical system:
\begin{align}
    \dot{d}_\tau &= d_\tau \: P_\tau(d,\pi) - d_\tau,  \; \forall \tau, \tag{EV.1} \label{eq:EvolutionState} \\
    \dot{\pi}_\tau[\cdot \mid u,k] &= \eta \: H^k(Q_\tau[u,k,\cdot](d,\pi),\pi_\tau[\cdot \mid u,k]), \; \forall \tau, u, k. \tag{EV.2} \label{eq:EvolutionPolicy}
\end{align}
See~\cite[Section~6]{elokda2021dynamic} for an interpretation. We use $\eta$ as a \emph{policy update rate} parameter, which controls the rate of policy changes~\eqref{eq:EvolutionPolicy} relative to the rate of state changes~\eqref{eq:EvolutionState}, and
\[
    H^k : \Real^{k+1} \times \Delta(\Bids^k) \rightarrow \Real^{k+1}
\]
is one of the familiar \emph{mean dynamics} in evolutionary game theory~\cite[Part~II]{sandholm2010population}. We choose the \emph{perturbed best response dynamic} as our mean dynamic, which allows modelling agents that are not perfectly rational.
The \emph{perturbed best response policy} is given by:
\begin{align}
    \label{eq:PerturbedBestResponse}
    \pbr_\tau[b \mid u,k](d,\pi) := \frac{\exp{(\lambda \: Q_\tau[u,k,b](d,\pi))}}{\sum_{b'} \exp{(\lambda \: Q_\tau[u,k,b'](d,\pi))}},
\end{align}
where the parameter $\lambda$ controls the degree of rationality of the agents.
For $\lambda = 0$, the perturbed best response is a uniform random distribution over all the bids.
At finite values of $\lambda$, it assigns higher probabilities to bids with higher single-stage deviation rewards in a smooth manner.
At the limit $\lambda \rightarrow \infty$, we recover the perfect best response with a uniform random distribution over all the bids maximizing the single-stage deviation rewards.
The policy update dynamics~\eqref{eq:EvolutionPolicy} are then simply
\begin{align}
    \dot{\pi}_\tau[\cdot \mid u,k] = \eta \: (\pbr_\tau[\cdot \mid u,k](d,\pi) - \pi_\tau[\cdot \mid u,k]), \; \forall u, k. \label{eq:PerturbedBestResponseEvolution}
\end{align}
Note that these dynamics lead to perturbed versions of the equilibrium policies $\epi_\tau$ at the rest points, rather than the exact policies~\cite[Section~6.2.4]{sandholm2010population}. However, for a sufficiently large $\lambda$, exact policies can be computed in practice, and we use $\lambda=1000$ in our computations\footnote{Numerical conditioning is required for high $\lambda$ to ensure that the exponential is stable. This is achieved with $\pbr_\tau[b \mid u,k](d,\pi) = \frac{\exp{(\lambda \: Q_\tau[u,k,b](d,\pi) - \max_{b^*} \lambda \: Q_\tau[u,k,b^*](d,\pi))}}{\sum_{b'} \exp{(\lambda \: Q_\tau[u,k,b'](d,\pi) - \max_{b^*} \lambda \: Q_\tau[u,k,b^*](d,\pi))}}$.}.

For computation purposes, we discretize the dynamics~\eqref{eq:EvolutionState}--\eqref{eq:EvolutionPolicy} using the discrete step size $dt$, yielding Algorithm~\ref{alg:StationaryEquilibrium}.

\begin{algorithm}[tb!]
\caption{Stationary Nash equilibrium computation}\label{alg:StationaryEquilibrium}
\textbf{Input:} Initial state distribution $d^0$ with average karma $\kbar$, initial policy $\pi^0$,  parameters $\eta$, $\lambda$, $dt$.

\textbf{Initialize:} $\sd \gets d^0$, $\epi \gets \pi^0$.

\textbf{While} $(\sd, \epi)$ not converged,
\begin{enumerate}
\item Compute:
\begin{enumerate}
\item $\bdist[b'](\sd,\epi)$, $\oprob[o \mid b](\sd,\epi)$, $\kappa[k^+ \mid k,b,o](\sd,\epi)$,
\item $\reward[u,b](\sd,\epi)$, $\transition_\tau[u^+,k^+ \mid u,k,b](\sd,\epi)$,
\item $R_\tau[u,k](\sd,\epi)$, $P_\tau[u^+,k^+ \mid u,k](\sd,\epi)$,
\item $V_\tau[u,k](\sd,\epi)$, $Q_\tau[u,k,b](\sd,\epi)$, and
\item $\pbr_\tau[b \mid u,k](\sd,\epi)$.
\end{enumerate}
\item Forward-propagate the discretized evolutionary dynamics:
\begin{align*}
    \sd_\tau &\gets (1 - dt) \: \sd_\tau + dt \: \sd_\tau \: P_\tau(\sd,\epi), \\
    \epi_\tau[\cdot \mid u,k] &\gets (1 - \eta \: dt) \: \epi_\tau[\cdot \mid u,k] + \eta \: dt \: \pbr_\tau[\cdot \mid u,k](\sd,\epi). \\
\end{align*}
\vspace*{-8mm}
\end{enumerate}
\end{algorithm}

% Figure~\ref{fig:lambda-sweep} shows example stationary Nash equilibrium policy computations for $\lambda \in \{10,100,1000\}$.
% The policies become increasingly deterministic with increasing $\lambda$. We choose $\lambda=1000$ in the subsequent analysis.

% \begin{figure}[bt!]
% 	\centering
% 	\includegraphics[height=.2\textwidth]{tikz/tikz-lambda-10.pdf}
% 	\hfil
% 	\includegraphics[height=.2\textwidth]{tikz/tikz-lambda-100.pdf}
% 	\hfil
% 	\includegraphics[height=.2\textwidth]{tikz/tikz-lambda-1000.pdf}
	
% 	\includegraphics[height=.2\textwidth]{tikz/tikz-lambda-10-karma-dist.pdf}
% 	\hfil
% 	\includegraphics[height=.2\textwidth]{tikz/tikz-lambda-100-karma-dist.pdf}
% 	\hfil
% 	\includegraphics[height=.2\textwidth]{tikz/tikz-lambda-1000-karma-dist.pdf}
	
% 	\caption{Stationary Nash equilibrium policy and karma distribution computations for different values of the rationality parameter $\lambda$. The plots are truncated at $k=20$ for display purposes; the rightmost value of the stationary karma distribution shows the proportion of agents with $k>=20$.}
% 	\label{fig:lambda-sweep}
% \end{figure}
\section{Stationary Nash Equilibrium in the case of no future discounting}
\label{sec:alpha1}

We have presented results for the case when agents do not discount their future rewards, i.e., when $\alpha_\tau=1$, which does not fit the dynamic population model of Section~\ref{sec:Model}. In particular, the expected infinite horizon reward $V_\tau(d,\pi)$~\eqref{eq:V-function-full} is not well defined for $\alpha_\tau = 1$.
Nevertheless, the case of no future discounting can be treated by considering that the agents face an \emph{average reward per time step} problem rather than a discounted problem~\cite{bertsekas2007dynamic}.
Analysis of the average reward per time step problem is intricate and we refer the reader to~\cite[Chapter~4]{bertsekas2007dynamic} for the details. For our purposes, it suffices to use the following proposition.

\begin{proposition}[\cite{bertsekas2007dynamic}~Propositions~4.2.1--4.2.2]
For a fixed social state $(d,\pi)$, if a scalar $\sigma_\tau(d,\pi)$ and a vector $Y_\tau(d,\pi)$ satisfy
\begin{multline}
    \label{eq:alpha1-V-function}
    \sigma_\tau(d,\pi) + Y_\tau[u,k](d,\pi) \\
    = R_\tau[u,k](d,\pi) + \sum_{u^+,k^+} P_\tau[u^+,k^+ \mid u,k](d,\pi) \: Y_\tau[u^+,k^+](d,\pi), \quad \forall u,k,
\end{multline}
then $\sigma_\tau(d,\pi)$ is the expected average reward per time step of the ego agent of type $\tau$ starting from any state $[u,k]$, supposing that $(d,\pi)$ is not time-varying.

Furthermore, if
\begin{multline}
    \label{eq:alpha1Optimality}
    \sigma_\tau(d,\pi) + Y_\tau[u,k](d,\pi) \\
    = \max_{b \in \Bids^k} \left[\reward[u,b](d,\pi) + \sum_{u^+,k^+} \transition_\tau[u^+,k^+ \mid u,k,b](d,\pi) \: Y_\tau[u^+,k^+](d,\pi)\right], \quad \forall u,k,
\end{multline}
then $\sigma_\tau(d,\pi)$ is the optimal expected average reward per time step of the ego agent of type $\tau$, and $\pi_\tau$ is an optimal policy for $\alpha_\tau=1$.
\end{proposition}

A stationary Nash equilibrium for $\alpha_\tau=1$ is then a social state $(\sd,\epi)$ where $\sd$ is stationary, i.e., \eqref{eq:SNE-1} holds, and $\epi_\tau$ satisfies~\eqref{eq:alpha1Optimality}.
Noting the similarity between~\eqref{eq:alpha1-V-function} and~\eqref{eq:V-function-full}, Algorithm~\ref{alg:StationaryEquilibrium} can be readily modified to seek a stationary Nash equilibrium for $\alpha_\tau=1$.
Instead of computing $V_\tau(d,\pi)$, we compute the pair $(\sigma_\tau(d,\pi),Y_\tau(d,\pi))$ as follows.
Noting that if $Y_\tau(d,\pi)$ satisfies~\eqref{eq:alpha1-V-function}, so will $\tilde{Y}_\tau(d,\pi) = Y_\tau(d,\pi) + y \mathbbm{1}$ for any constant offset $y$, we can eliminate this degree of freedom by fixing $Y_\tau[u_1,0](d,\pi) = 0$ for the arbitrary `default state' $[u,k]=[u_1,0]$, which allows us to re-write~\eqref{eq:alpha1-V-function} as
\begin{align*}
    Y_\tau[u_1,0](d,\pi) &= 0, \\
    \sigma_\tau(d,\pi) &= R_\tau[u_1,0](d,\pi) + \sum_{u^+,k^+} P_\tau[u^+,k^+ \mid u_1,0](d,\pi) \: Y_\tau[u^+,k^+](d,\pi), \\
    Y_\tau[u,k](d,\pi) &= R_\tau[u,k](d,\pi) + \sum_{u^+,k^+} P_\tau[u^+,k^+ \mid u,k](d,\pi) \: Y_\tau[u^+,k^+](d,\pi)\\
    &\phantom{=} \quad  - \sigma_\tau(d,\pi), \quad \forall u,k.
\end{align*}
This system of equations can be solved iteratively by using an initial guess $Y_\tau^0(d,\pi)$ and updating $\sigma_\tau(d,\pi)$ and $Y_\tau(d,\pi)$ until convergence, a procedure known as \emph{relative value iteration}~\cite[Section~4.3.1]{bertsekas2007dynamic}.
The vector $Y_\tau(d,\pi)$ can be interpreted as a relative rewards vector, with $Y_\tau[u,k](d,\pi)$ representing transient reward advantages/disadvantages of state $[u,k]$ with respect to the `default state' $[u_1,0]$.
With this, we can define the single-stage deviation rewards $Q_\tau[u,k,b](d,\pi)$~\eqref{eq:SingleStageDeviation} with respect to $Y_\tau(d,\pi)$ instead of $V_\tau(d,\pi)$, and proceed with the remainder of Algorithm~\ref{alg:StationaryEquilibrium} unmodified.

\begin{figure*}[tb!]
    \centering
% 	\includegraphics[height=.2\textwidth]{tikz/tikz-PBS-alpha-1-policy.pdf}
% 	\hfil
% 	\includegraphics[height=.2\textwidth]{tikz/tikz-PBS-alpha-1-karma-dist.pdf}
% 	\medskip
% 	\hrule
	
	\includegraphics[trim=27.95416pt 0 16.01173pt 0]{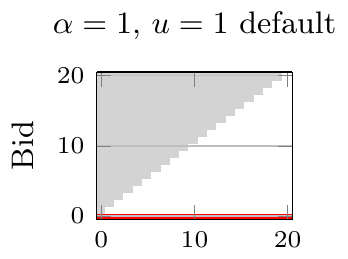}
	\hfil
	\includegraphics[trim=16.53604pt 0 16.53604pt 0]{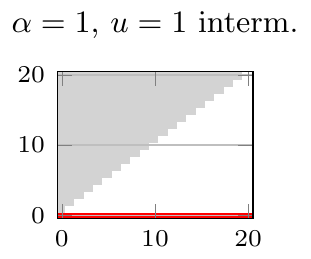}
	\hfil
	\includegraphics[trim=17.30687pt 0 17.30687pt 0]{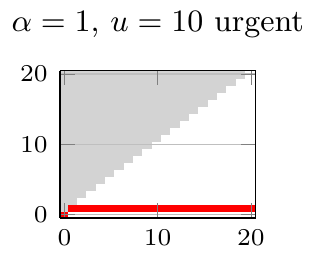}
	
	\includegraphics[trim=27.95416pt 0 10.86873pt 0]{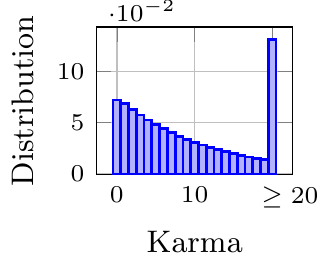}
	\hfil
	\includegraphics[trim=17.21321pt 0 10.86873pt 0]{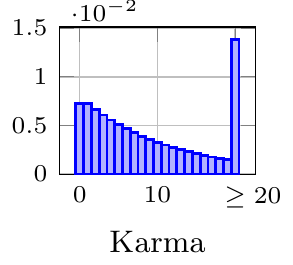}
	\hfil
	\includegraphics[trim=17.21321pt 0 10.86873pt 0]{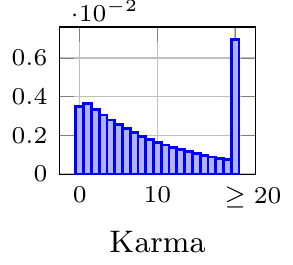}
	
	\caption{Stationary Nash equilibrium when agents do not discount the future ($\alpha=1$).}
	\label{fig:PBS-alpha-1}
\end{figure*}

Figure~\ref{fig:PBS-alpha-1} shows the stationary Nash equilibrium under karma payment rule $\texttt{PBS}$ when all agents have $\alpha=1$ (we drop the subscript $\tau$ here because there is only one type), computed using this method.
The defining feature of the equilibrium bidding policy is that it is very sparing on karma, with the agents bidding 1 only in the high urgency state and zero otherwise, no matter how much karma they have (unless they have zero karma).
While on first investigation, this behaviour resembles a truthful declaration of the urgency, and indeed $\alpha=1$ performs well for the efficiency, it has a few negative consequences.
First, it leads to a high stationary mass of the high urgency, zero karma state, which means that a high fraction of times the high urgency agents do not have one karma to bid.
Second, the homogeneity of the bids leads to the karma losing its effectiveness as a turn-taking device, and results in the poor ex-post access fairness observed in Figure~\ref{fig:U-1-1-10-PBP-PBS-performance} at $\alpha=1$.

% Consequently, the \emph{karma hoarding} behavior that emerges when the agents do not discount the future is hurtful for the social welfare, both in terms of the efficiency and the fairness.
As we explored in Section~\ref{subsec:HeterogeneousAlpha}, an effective measure against this \emph{karma hoarding} behavior is to let the karma expire, e.g., by applying a karma tax.




\end{appendices}


\bibliography{bibliography}% common bib file


\begin{thebibliography}{60}
\ifx \bisbn   \undefined \def \bisbn  #1{ISBN #1}\fi
\ifx \binits  \undefined \def \binits#1{#1}\fi
\ifx \bauthor  \undefined \def \bauthor#1{#1}\fi
\ifx \batitle  \undefined \def \batitle#1{#1}\fi
\ifx \bjtitle  \undefined \def \bjtitle#1{#1}\fi
\ifx \bvolume  \undefined \def \bvolume#1{\textbf{#1}}\fi
\ifx \byear  \undefined \def \byear#1{#1}\fi
\ifx \bissue  \undefined \def \bissue#1{#1}\fi
\ifx \bfpage  \undefined \def \bfpage#1{#1}\fi
\ifx \blpage  \undefined \def \blpage #1{#1}\fi
\ifx \burl  \undefined \def \burl#1{\textsf{#1}}\fi
\ifx \doiurl  \undefined \def \doiurl#1{\url{https://doi.org/#1}}\fi
\ifx \betal  \undefined \def \betal{\textit{et al.}}\fi
\ifx \binstitute  \undefined \def \binstitute#1{#1}\fi
\ifx \binstitutionaled  \undefined \def \binstitutionaled#1{#1}\fi
\ifx \bctitle  \undefined \def \bctitle#1{#1}\fi
\ifx \beditor  \undefined \def \beditor#1{#1}\fi
\ifx \bpublisher  \undefined \def \bpublisher#1{#1}\fi
\ifx \bbtitle  \undefined \def \bbtitle#1{#1}\fi
\ifx \bedition  \undefined \def \bedition#1{#1}\fi
\ifx \bseriesno  \undefined \def \bseriesno#1{#1}\fi
\ifx \blocation  \undefined \def \blocation#1{#1}\fi
\ifx \bsertitle  \undefined \def \bsertitle#1{#1}\fi
\ifx \bsnm \undefined \def \bsnm#1{#1}\fi
\ifx \bsuffix \undefined \def \bsuffix#1{#1}\fi
\ifx \bparticle \undefined \def \bparticle#1{#1}\fi
\ifx \barticle \undefined \def \barticle#1{#1}\fi
\bibcommenthead
\ifx \bconfdate \undefined \def \bconfdate #1{#1}\fi
\ifx \botherref \undefined \def \botherref #1{#1}\fi
\ifx \url \undefined \def \url#1{\textsf{#1}}\fi
\ifx \bchapter \undefined \def \bchapter#1{#1}\fi
\ifx \bbook \undefined \def \bbook#1{#1}\fi
\ifx \bcomment \undefined \def \bcomment#1{#1}\fi
\ifx \oauthor \undefined \def \oauthor#1{#1}\fi
\ifx \citeauthoryear \undefined \def \citeauthoryear#1{#1}\fi
\ifx \endbibitem  \undefined \def \endbibitem {}\fi
\ifx \bconflocation  \undefined \def \bconflocation#1{#1}\fi
\ifx \arxivurl  \undefined \def \arxivurl#1{\textsf{#1}}\fi
\csname PreBibitemsHook\endcsname

%%% 1
\bibitem{ostrom1990governing}
\begin{bbook}
\bauthor{\bsnm{Ostrom}, \binits{E.}}:
\bbtitle{Governing the Commons: The Evolution of Institutions for Collective
  Action}.
\bpublisher{Cambridge University Press},
\blocation{Cambridge}
(\byear{1990})
\end{bbook}
\endbibitem

%%% 2
\bibitem{berkes1986local}
\begin{barticle}
\bauthor{\bsnm{Berkes}, \binits{F.}}:
\batitle{Local-level management and the commons problem: A comparative study of
  turkish coastal fisheries}.
\bjtitle{Marine Policy}
\bvolume{10}(\bissue{3}),
\bfpage{215}--\blpage{229}
(\byear{1986})
\end{barticle}
\endbibitem

%%% 3
\bibitem{censi2019today}
\begin{bchapter}
\bauthor{\bsnm{Censi}, \binits{A.}},
\bauthor{\bsnm{Bolognani}, \binits{S.}},
\bauthor{\bsnm{Zilly}, \binits{J.G.}},
\bauthor{\bsnm{Mousavi}, \binits{S.S.}},
\bauthor{\bsnm{Frazzoli}, \binits{E.}}:
\bctitle{Today me, tomorrow thee: Efficient resource allocation in competitive
  settings using karma games}.
In: \bbtitle{IEEE Intelligent Transportation Systems Conference (ITSC)},
pp. \bfpage{686}--\blpage{693}
(\byear{2019})
\end{bchapter}
\endbibitem

%%% 4
\bibitem{doniger1980karma}
\begin{bbook}
\bauthor{\bsnm{O'Flaherty}, \binits{W.D.}}:
\bbtitle{Karma and Rebirth in Classical Indian Traditions}.
\bpublisher{University of California Press},
\blocation{Berkeley CA}
(\byear{1980})
\end{bbook}
\endbibitem

%%% 5
\bibitem{castillo2017surge}
\begin{bchapter}
\bauthor{\bsnm{Castillo}, \binits{J.C.}},
\bauthor{\bsnm{Knoepfle}, \binits{D.}},
\bauthor{\bsnm{Weyl}, \binits{G.}}:
\bctitle{Surge pricing solves the wild goose chase}.
In: \bbtitle{ACM Conference on Economics and Computation},
pp. \bfpage{241}--\blpage{242}
(\byear{2017})
\end{bchapter}
\endbibitem

%%% 6
\bibitem{cachon2017role}
\begin{barticle}
\bauthor{\bsnm{Cachon}, \binits{G.P.}},
\bauthor{\bsnm{Daniels}, \binits{K.M.}},
\bauthor{\bsnm{Lobel}, \binits{R.}}:
\batitle{The role of surge pricing on a service platform with self-scheduling
  capacity}.
\bjtitle{Manufacturing \& Service Operations Management}
\bvolume{19}(\bissue{3}),
\bfpage{368}--\blpage{384}
(\byear{2017})
\end{barticle}
\endbibitem

%%% 7
\bibitem{dholakia2015everyone}
\begin{botherref}
\oauthor{\bsnm{Dholakia}, \binits{U.M.}}:
Everyone hates Uber's surge pricing – here's how to fix it
(2015).
\url{https://hbr.org/2015/12/everyone-hates-ubers-surge-pricing-heres-how-to-fix-it}
\end{botherref}
\endbibitem

%%% 8
\bibitem{shontell2014uber}
\begin{botherref}
\oauthor{\bsnm{Shontell}, \binits{A.}},
\oauthor{\bsnm{Moss}, \binits{C.}}:
Uber denies text message that suggests it withheld cars `to make earnings
  higher' on Valentine's Day
(2014).
\url{https://www.businessinsider.com/uber-surge-pricing-text-2014-2?r=US&IR=T}
\end{botherref}
\endbibitem

%%% 9
\bibitem{yang2011managing}
\begin{barticle}
\bauthor{\bsnm{Yang}, \binits{H.}},
\bauthor{\bsnm{Wang}, \binits{X.}}:
\batitle{Managing network mobility with tradable credits}.
\bjtitle{Transportation Research Part B: Methodological}
\bvolume{45}(\bissue{3}),
\bfpage{580}--\blpage{594}
(\byear{2011})
\end{barticle}
\endbibitem

%%% 10
\bibitem{xiao2013managing}
\begin{barticle}
\bauthor{\bsnm{Xiao}, \binits{F.}},
\bauthor{\bsnm{Qian}, \binits{Z.S.}},
\bauthor{\bsnm{Zhang}, \binits{H.M.}}:
\batitle{Managing bottleneck congestion with tradable credits}.
\bjtitle{Transportation Research Part B: Methodological}
\bvolume{56},
\bfpage{1}--\blpage{14}
(\byear{2013})
\end{barticle}
\endbibitem

%%% 11
\bibitem{brands2020tradable}
\begin{barticle}
\bauthor{\bsnm{Brands}, \binits{D.K.}},
\bauthor{\bsnm{Verhoef}, \binits{E.T.}},
\bauthor{\bsnm{Knockaert}, \binits{J.}},
\bauthor{\bsnm{Koster}, \binits{P.R.}}:
\batitle{Tradable permits to manage urban mobility: market design and
  experimental implementation}.
\bjtitle{Transportation Research Part A: Policy and Practice}
\bvolume{137},
\bfpage{34}--\blpage{46}
(\byear{2020})
\end{barticle}
\endbibitem

%%% 12
\bibitem{borjesson2012income}
\begin{barticle}
\bauthor{\bsnm{B{\"o}rjesson}, \binits{M.}},
\bauthor{\bsnm{Fosgerau}, \binits{M.}},
\bauthor{\bsnm{Algers}, \binits{S.}}:
\batitle{On the income elasticity of the value of travel time}.
\bjtitle{Transportation Research Part A: Policy and Practice}
\bvolume{46}(\bissue{2}),
\bfpage{368}--\blpage{377}
(\byear{2012})
\end{barticle}
\endbibitem

%%% 13
\bibitem{obama2016net}
\begin{botherref}
Net Neutrality - President Obama's plan for a free and open internet
(2016).
\url{https://obamawhitehouse.archives.gov/net-neutrality}
\end{botherref}
\endbibitem

%%% 14
\bibitem{bourreau2015net}
\begin{barticle}
\bauthor{\bsnm{Bourreau}, \binits{M.}},
\bauthor{\bsnm{Kourandi}, \binits{F.}},
\bauthor{\bsnm{Valletti}, \binits{T.}}:
\batitle{Net neutrality with competing internet platforms}.
\bjtitle{The Journal of Industrial Economics}
\bvolume{63}(\bissue{1}),
\bfpage{30}--\blpage{73}
(\byear{2015})
\end{barticle}
\endbibitem

%%% 15
\bibitem{pil2010net}
\begin{barticle}
\bauthor{\bsnm{Pil~Choi}, \binits{J.}},
\bauthor{\bsnm{Kim}, \binits{B.-C.}}:
\batitle{Net neutrality and investment incentives}.
\bjtitle{The RAND Journal of Economics}
\bvolume{41}(\bissue{3}),
\bfpage{446}--\blpage{471}
(\byear{2010})
\end{barticle}
\endbibitem

%%% 16
\bibitem{hahn2006economics}
\begin{botherref}
\oauthor{\bsnm{Hahn}, \binits{R.W.}},
\oauthor{\bsnm{Wallsten}, \binits{S.}}:
The economics of net neutrality.
The Economists' Voice
\textbf{3}(6)
(2006)
\end{botherref}
\endbibitem

%%% 17
\bibitem{elokda2021dynamic}
\begin{botherref}
\oauthor{\bsnm{Elokda}, \binits{E.}},
\oauthor{\bsnm{Censi}, \binits{A.}},
\oauthor{\bsnm{Bolognani}, \binits{S.}}:
Dynamic population games.
arXiv preprint arXiv:2104.14662
(2021)
\end{botherref}
\endbibitem

%%% 18
\bibitem{friedman1971non}
\begin{barticle}
\bauthor{\bsnm{Friedman}, \binits{J.W.}}:
\batitle{A non-cooperative equilibrium for supergames}.
\bjtitle{The Review of Economic Studies}
\bvolume{38}(\bissue{1}),
\bfpage{1}--\blpage{12}
(\byear{1971})
\end{barticle}
\endbibitem

%%% 19
\bibitem{fudenberg1986folk}
\begin{barticle}
\bauthor{\bsnm{Fudenberg}, \binits{D.}},
\bauthor{\bsnm{Maskin}, \binits{E.}}:
\batitle{The folk theorem in repeated games with discounting or with incomplete
  information}.
\bjtitle{Econometrica}
\bvolume{54}(\bissue{3}),
\bfpage{533}--\blpage{554}
(\byear{1986})
\end{barticle}
\endbibitem

%%% 20
\bibitem{fudenberg1994folk}
\begin{barticle}
\bauthor{\bsnm{Fudenberg}, \binits{D.}},
\bauthor{\bsnm{Levine}, \binits{D.K.}},
\bauthor{\bsnm{Maskin}, \binits{E.}}:
\batitle{The folk theorem with imperfect public information}.
\bjtitle{Econometrica}
\bvolume{62}(\bissue{5}),
\bfpage{997}--\blpage{1039}
(\byear{1994})
\end{barticle}
\endbibitem

%%% 21
\bibitem{fudenberg1991approximate}
\begin{barticle}
\bauthor{\bsnm{Fudenberg}, \binits{D.}},
\bauthor{\bsnm{Levine}, \binits{D.K.}}:
\batitle{An approximate folk theorem with imperfect private information}.
\bjtitle{Journal of Economic Theory}
\bvolume{54}(\bissue{1}),
\bfpage{26}--\blpage{47}
(\byear{1991})
\end{barticle}
\endbibitem

%%% 22
\bibitem{dutta1995folk}
\begin{barticle}
\bauthor{\bsnm{Dutta}, \binits{P.K.}}:
\batitle{A folk theorem for stochastic games}.
\bjtitle{Journal of Economic Theory}
\bvolume{66}(\bissue{1}),
\bfpage{1}--\blpage{32}
(\byear{1995})
\end{barticle}
\endbibitem

%%% 23
\bibitem{okuno1995social}
\begin{barticle}
\bauthor{\bsnm{Okuno-Fujiwara}, \binits{M.}},
\bauthor{\bsnm{Postlewaite}, \binits{A.}}:
\batitle{Social norms and random matching games}.
\bjtitle{Games and Economic behavior}
\bvolume{9}(\bissue{1}),
\bfpage{79}--\blpage{109}
(\byear{1995})
\end{barticle}
\endbibitem

%%% 24
\bibitem{jackson2007overcoming}
\begin{barticle}
\bauthor{\bsnm{Jackson}, \binits{M.O.}},
\bauthor{\bsnm{Sonnenschein}, \binits{H.F.}}:
\batitle{Overcoming incentive constraints by linking decisions}.
\bjtitle{Econometrica}
\bvolume{75}(\bissue{1}),
\bfpage{241}--\blpage{257}
(\byear{2007})
\end{barticle}
\endbibitem

%%% 25
\bibitem{gibbard1973manipulation}
\begin{barticle}
\bauthor{\bsnm{Gibbard}, \binits{A.}}:
\batitle{Manipulation of voting schemes: a general result}.
\bjtitle{Econometrica}
\bvolume{41}(\bissue{4}),
\bfpage{587}--\blpage{601}
(\byear{1973})
\end{barticle}
\endbibitem

%%% 26
\bibitem{satterthwaite1975strategy}
\begin{barticle}
\bauthor{\bsnm{Satterthwaite}, \binits{M.A.}}:
\batitle{Strategy-proofness and {Arrow}'s conditions: Existence and
  correspondence theorems for voting procedures and social welfare functions}.
\bjtitle{Journal of Economic Theory}
\bvolume{10}(\bissue{2}),
\bfpage{187}--\blpage{217}
(\byear{1975})
\end{barticle}
\endbibitem

%%% 27
\bibitem{vickrey1961counterspeculation}
\begin{barticle}
\bauthor{\bsnm{Vickrey}, \binits{W.}}:
\batitle{Counterspeculation, auctions, and competitive sealed tenders}.
\bjtitle{The Journal of Finance}
\bvolume{16}(\bissue{1}),
\bfpage{8}--\blpage{37}
(\byear{1961})
\end{barticle}
\endbibitem

%%% 28
\bibitem{clarke1971multipart}
\begin{barticle}
\bauthor{\bsnm{Clarke}, \binits{E.H.}}:
\batitle{Multipart pricing of public goods}.
\bjtitle{Public Choice}
\bvolume{11}(\bissue{1}),
\bfpage{17}--\blpage{33}
(\byear{1971})
\end{barticle}
\endbibitem

%%% 29
\bibitem{groves1973incentives}
\begin{barticle}
\bauthor{\bsnm{Groves}, \binits{T.}}:
\batitle{Incentives in teams}.
\bjtitle{Econometrica}
\bvolume{41}(\bissue{4}),
\bfpage{617}--\blpage{631}
(\byear{1973})
\end{barticle}
\endbibitem

%%% 30
\bibitem{schummer2007mechanism}
\begin{bchapter}
\bauthor{\bsnm{Schummer}, \binits{J.}},
\bauthor{\bsnm{Vohra}, \binits{R.V.}}:
\bctitle{Mechanism design without money}.
In: \bbtitle{Algorithmic Game Theory},
pp. \bfpage{243}--\blpage{299}
(\byear{2007})
\end{bchapter}
\endbibitem

%%% 31
\bibitem{moulin1980strategy}
\begin{barticle}
\bauthor{\bsnm{Moulin}, \binits{H.}}:
\batitle{On strategy-proofness and single peakedness}.
\bjtitle{Public Choice}
\bvolume{35}(\bissue{4}),
\bfpage{437}--\blpage{455}
(\byear{1980})
\end{barticle}
\endbibitem

%%% 32
\bibitem{shapley1974cores}
\begin{barticle}
\bauthor{\bsnm{Shapley}, \binits{L.}},
\bauthor{\bsnm{Scarf}, \binits{H.}}:
\batitle{On cores and indivisibility}.
\bjtitle{Journal of Mathematical Economics}
\bvolume{1}(\bissue{1}),
\bfpage{23}--\blpage{37}
(\byear{1974})
\end{barticle}
\endbibitem

%%% 33
\bibitem{gale1962college}
\begin{barticle}
\bauthor{\bsnm{Gale}, \binits{D.}},
\bauthor{\bsnm{Shapley}, \binits{L.S.}}:
\batitle{College admissions and the stability of marriage}.
\bjtitle{The American Mathematical Monthly}
\bvolume{69}(\bissue{1}),
\bfpage{9}--\blpage{15}
(\byear{1962})
\end{barticle}
\endbibitem

%%% 34
\bibitem{hylland1979efficient}
\begin{barticle}
\bauthor{\bsnm{Hylland}, \binits{A.}},
\bauthor{\bsnm{Zeckhauser}, \binits{R.}}:
\batitle{The efficient allocation of individuals to positions}.
\bjtitle{Journal of Political Economy}
\bvolume{87}(\bissue{2}),
\bfpage{293}--\blpage{314}
(\byear{1979})
\end{barticle}
\endbibitem

%%% 35
\bibitem{sonmez2010course}
\begin{barticle}
\bauthor{\bsnm{S{\"o}nmez}, \binits{T.}},
\bauthor{\bsnm{{\"U}nver}, \binits{M.U.}}:
\batitle{Course bidding at business schools}.
\bjtitle{International Economic Review}
\bvolume{51}(\bissue{1}),
\bfpage{99}--\blpage{123}
(\byear{2010})
\end{barticle}
\endbibitem

%%% 36
\bibitem{budish2011combinatorial}
\begin{barticle}
\bauthor{\bsnm{Budish}, \binits{E.}}:
\batitle{The combinatorial assignment problem: Approximate competitive
  equilibrium from equal incomes}.
\bjtitle{Journal of Political Economy}
\bvolume{119}(\bissue{6}),
\bfpage{1061}--\blpage{1103}
(\byear{2011})
\end{barticle}
\endbibitem

%%% 37
\bibitem{budish2012multi}
\begin{barticle}
\bauthor{\bsnm{Budish}, \binits{E.}},
\bauthor{\bsnm{Cantillon}, \binits{E.}}:
\batitle{The multi-unit assignment problem: Theory and evidence from course
  allocation at harvard}.
\bjtitle{American Economic Review}
\bvolume{102}(\bissue{5}),
\bfpage{2237}--\blpage{71}
(\byear{2012})
\end{barticle}
\endbibitem

%%% 38
\bibitem{spear1987repeated}
\begin{barticle}
\bauthor{\bsnm{Spear}, \binits{S.E.}},
\bauthor{\bsnm{Srivastava}, \binits{S.}}:
\batitle{On repeated moral hazard with discounting}.
\bjtitle{The Review of Economic Studies}
\bvolume{54}(\bissue{4}),
\bfpage{599}--\blpage{617}
(\byear{1987})
\end{barticle}
\endbibitem

%%% 39
\bibitem{guo2020dynamic}
\begin{botherref}
\oauthor{\bsnm{Guo}, \binits{Y.}},
\oauthor{\bsnm{H{\"o}rner}, \binits{J.}}:
Dynamic allocation without money.
TSE Working Paper
(2020)
\end{botherref}
\endbibitem

%%% 40
\bibitem{balseiro2019multiagent}
\begin{barticle}
\bauthor{\bsnm{Balseiro}, \binits{S.R.}},
\bauthor{\bsnm{Gurkan}, \binits{H.}},
\bauthor{\bsnm{Sun}, \binits{P.}}:
\batitle{Multiagent mechanism design without money}.
\bjtitle{Operations Research}
\bvolume{67}(\bissue{5}),
\bfpage{1417}--\blpage{1436}
(\byear{2019})
\end{barticle}
\endbibitem

%%% 41
\bibitem{sonmez2020incentivized}
\begin{barticle}
\bauthor{\bsnm{S{\"o}nmez}, \binits{T.}},
\bauthor{\bsnm{{\"U}nver}, \binits{M.U.}},
\bauthor{\bsnm{Yenmez}, \binits{M.B.}}:
\batitle{Incentivized kidney exchange}.
\bjtitle{American Economic Review}
\bvolume{110}(\bissue{7}),
\bfpage{2198}--\blpage{2224}
(\byear{2020})
\end{barticle}
\endbibitem

%%% 42
\bibitem{kim2021organ}
\begin{barticle}
\bauthor{\bsnm{Kim}, \binits{J.}},
\bauthor{\bsnm{Li}, \binits{M.}},
\bauthor{\bsnm{Xu}, \binits{M.}}:
\batitle{Organ donation with vouchers}.
\bjtitle{Journal of Economic Theory}
\bvolume{191},
\bfpage{105}--\blpage{159}
(\byear{2021})
\end{barticle}
\endbibitem

%%% 43
\bibitem{golle2001incentives}
\begin{bchapter}
\bauthor{\bsnm{Golle}, \binits{P.}},
\bauthor{\bsnm{Leyton-Brown}, \binits{K.}},
\bauthor{\bsnm{Mironov}, \binits{I.}},
\bauthor{\bsnm{Lillibridge}, \binits{M.}}:
\bctitle{Incentives for sharing in peer-to-peer networks}.
In: \bbtitle{International Workshop on Electronic Commerce},
pp. \bfpage{75}--\blpage{87}
(\byear{2001}).
\bcomment{Springer}
\end{bchapter}
\endbibitem

%%% 44
\bibitem{vishnumurthy2003karma}
\begin{bchapter}
\bauthor{\bsnm{Vishnumurthy}, \binits{V.}},
\bauthor{\bsnm{Chandrakumar}, \binits{S.}},
\bauthor{\bsnm{Sirer}, \binits{E.G.}}:
\bctitle{Karma: A secure economic framework for peer-to-peer resource sharing}.
In: \bbtitle{Workshop on Economics of Peer-to-peer Systems}
(\byear{2003})
\end{bchapter}
\endbibitem

%%% 45
\bibitem{friedman2006efficiency}
\begin{bchapter}
\bauthor{\bsnm{Friedman}, \binits{E.J.}},
\bauthor{\bsnm{Halpern}, \binits{J.Y.}},
\bauthor{\bsnm{Kash}, \binits{I.}}:
\bctitle{Efficiency and {Nash} equilibria in a scrip system for {P2P}
  networks}.
In: \bbtitle{ACM Conference on Electronic Commerce},
pp. \bfpage{140}--\blpage{149}
(\byear{2006})
\end{bchapter}
\endbibitem

%%% 46
\bibitem{salazar2021urgency}
\begin{barticle}
\bauthor{\bsnm{Salazar}, \binits{M.}},
\bauthor{\bsnm{Paccagnan}, \binits{D.}},
\bauthor{\bsnm{Agazzi}, \binits{A.}},
\bauthor{\bsnm{Heemels}, \binits{W.M.}}:
\batitle{Urgency-aware optimal routing in repeated games through artificial
  currencies}.
\bjtitle{European Journal of Control}
\bvolume{62},
\bfpage{22}--\blpage{32}
(\byear{2021})
\end{barticle}
\endbibitem

%%% 47
\bibitem{prendergast2022allocation}
\begin{barticle}
\bauthor{\bsnm{Prendergast}, \binits{C.}}:
\batitle{The allocation of food to food banks}.
\bjtitle{Journal of Political Economy}
\bvolume{130}(\bissue{8}),
\bfpage{1993}--\blpage{2017}
(\byear{2022})
\end{barticle}
\endbibitem

%%% 48
\bibitem{johnson2014analyzing}
\begin{barticle}
\bauthor{\bsnm{Johnson}, \binits{K.}},
\bauthor{\bsnm{Simchi-Levi}, \binits{D.}},
\bauthor{\bsnm{Sun}, \binits{P.}}:
\batitle{Analyzing scrip systems}.
\bjtitle{Operations Research}
\bvolume{62}(\bissue{3}),
\bfpage{524}--\blpage{534}
(\byear{2014})
\end{barticle}
\endbibitem

%%% 49
\bibitem{gorokh2021monetary}
\begin{barticle}
\bauthor{\bsnm{Gorokh}, \binits{A.}},
\bauthor{\bsnm{Banerjee}, \binits{S.}},
\bauthor{\bsnm{Iyer}, \binits{K.}}:
\batitle{From monetary to nonmonetary mechanism design via artificial
  currencies}.
\bjtitle{Mathematics of Operations Research}
\bvolume{46}(\bissue{3}),
\bfpage{835}--\blpage{855}
(\byear{2021})
\end{barticle}
\endbibitem

%%% 50
\bibitem{sandholm2010population}
\begin{bbook}
\bauthor{\bsnm{Sandholm}, \binits{W.H.}}:
\bbtitle{Population Games and Evolutionary Dynamics}.
\bpublisher{MIT Press},
\blocation{Cambridge MA}
(\byear{2010})
\end{bbook}
\endbibitem

%%% 51
\bibitem{treves1967topological}
\begin{bbook}
\bauthor{\bsnm{Treves}, \binits{F.}}:
\bbtitle{Topological Vector Spaces, Distributions and Kernels}.
\bpublisher{Academic Press},
\blocation{San Diego CA}
(\byear{1967})
\end{bbook}
\endbibitem

%%% 52
\bibitem{rudin1976principles}
\begin{bbook}
\bauthor{\bsnm{Rudin}, \binits{W.}}:
\bbtitle{Principles of Mathematical Analysis}
vol. \bseriesno{3}.
\bpublisher{McGraw-Hill},
\blocation{New York NY}
(\byear{1976})
\end{bbook}
\endbibitem

%%% 53
\bibitem{granas2003fixed}
\begin{bbook}
\bauthor{\bsnm{Granas}, \binits{A.}},
\bauthor{\bsnm{James}, \binits{D.}}:
\bbtitle{Fixed Point Theory}.
\bpublisher{Springer},
\blocation{New York NY}
(\byear{2003})
\end{bbook}
\endbibitem

%%% 54
\bibitem{tychonoff1930topologische}
\begin{barticle}
\bauthor{\bsnm{Tychonoff}, \binits{A.}}:
\batitle{{\"U}ber die topologische {Erweiterung} von {R{\"a}umen}}.
\bjtitle{Mathematische Annalen}
\bvolume{102}(\bissue{1}),
\bfpage{544}--\blpage{561}
(\byear{1930})
\end{barticle}
\endbibitem

%%% 55
\bibitem{berge1997topological}
\begin{bbook}
\bauthor{\bsnm{Berge}, \binits{C.}}:
\bbtitle{Topological Spaces: Including a Treatment of Multi-valued Functions,
  Vector Spaces, and Convexity}.
\bpublisher{Dover Publications},
\blocation{Mineola NY}
(\byear{1997})
\end{bbook}
\endbibitem

%%% 56
\bibitem{nisan2007algorithmic}
\begin{bbook}
\bauthor{\bsnm{Nisan}, \binits{N.}},
\bauthor{\bsnm{Roughgarden}, \binits{T.}},
\bauthor{\bsnm{Tardos}, \binits{E.}},
\bauthor{\bsnm{Vazirani}, \binits{V.V.}}:
\bbtitle{Algorithmic Game Theory}.
\bpublisher{Cambridge University Press},
\blocation{Cambridge}
(\byear{2007})
\end{bbook}
\endbibitem

%%% 57
\bibitem{krishna2009auction}
\begin{bbook}
\bauthor{\bsnm{Krishna}, \binits{V.}}:
\bbtitle{Auction Theory}.
\bpublisher{Academic Press},
\blocation{Burlington MA}
(\byear{2009})
\end{bbook}
\endbibitem

%%% 58
\bibitem{cappelen2013just}
\begin{barticle}
\bauthor{\bsnm{Cappelen}, \binits{A.W.}},
\bauthor{\bsnm{Konow}, \binits{J.}},
\bauthor{\bsnm{S{\o}rensen}, \binits{E.{\O}.}},
\bauthor{\bsnm{Tungodden}, \binits{B.}}:
\batitle{Just luck: An experimental study of risk-taking and fairness}.
\bjtitle{American Economic Review}
\bvolume{103}(\bissue{4}),
\bfpage{1398}--\blpage{1413}
(\byear{2013})
\end{barticle}
\endbibitem

%%% 59
\bibitem{roberts1980interpersonal}
\begin{barticle}
\bauthor{\bsnm{Roberts}, \binits{K.W.}}:
\batitle{Interpersonal comparability and social choice theory}.
\bjtitle{The Review of Economic Studies}
\bvolume{47}(\bissue{2}),
\bfpage{421}--\blpage{439}
(\byear{1980})
\end{barticle}
\endbibitem

%%% 60
\bibitem{bertsekas2007dynamic}
\begin{bbook}
\bauthor{\bsnm{Bertsekas}, \binits{D.}}:
\bbtitle{Dynamic Programming and Optimal Control}
vol. \bseriesno{2}.
\bpublisher{Athena Scientific},
\blocation{Belmont MA}
(\byear{2007})
\end{bbook}
\endbibitem

\end{thebibliography}


\end{document}